\documentclass[11pt]{article}

\usepackage[letterpaper,margin=1in]{geometry}

\usepackage[utf8]{inputenc}
\usepackage[T1]{fontenc}
\usepackage{microtype}
\usepackage{amsmath,amssymb,amsthm}
\usepackage{graphicx}
\usepackage{float}
\usepackage{booktabs}
\usepackage{placeins}
\usepackage{tabularx}
\usepackage[labelfont=bf]{caption}
\usepackage[numbers,sort&compress]{natbib}
\usepackage{xcolor}
\usepackage{hyperref}
\usepackage{orcidlink}
\usepackage{xurl}
\hypersetup{colorlinks=true, linkcolor=blue!60!black,
            citecolor=blue!60!black, urlcolor=blue!60!black}

\graphicspath{{../repo/figures/}{figures/}}
 \theoremstyle{plain}
\newtheorem{theorem}{Theorem}[section]
\newtheorem{proposition}[theorem]{Proposition}
\newtheorem{lemma}[theorem]{Lemma}
\newtheorem{conjecture}[theorem]{Conjecture}
 \newcommand{\Ee}{\mathbb{E}}
\newcommand{\ket}[1]{|#1\rangle}
\newcommand{\bra}[1]{\langle #1|}
\newcommand{\braket}[2]{\langle #1 | #2 \rangle}
\newcommand{\Tr}{\mathrm{Tr}}
 
\begin{document}

\hypersetup{pageanchor=false}
\title{The optimization landscape of peaked-circuit generation}

\author{Ilyes Jamoussi\,\orcidlink{0009-0008-9098-8329}\thanks{ilyes.jamoussi@polymtl.ca} \\[0.3em]
Polytechnique Montr\'eal, Montr\'eal, Qu\'ebec, Canada}

\date{\today}

\maketitle
\thispagestyle{empty}
 \begin{abstract}
Peaked circuits are random quantum circuits whose measurement
returns one bitstring far more often than chance. The probability
of that string is the peakedness, and a verifier who knows the
string checks the device in a few shots. Peaked circuits are a
candidate route to verifiable quantum advantage, and classical
generation is the bottleneck. Aaronson and Zhang train a
variational circuit by gradient descent, and the peakedness they
reach decays exponentially with the number of qubits. They state
that either their optimizer stalls or no efficient generation
method exists. Here we show that no bare fixed base fits the decay we measure and
that neither the barren plateau nor fragmentation of the
solution set explains it. We measure the
optimization landscape on eighteen instances per size, from eight
to sixteen qubits, under a protocol whose falsifiers were fixed in
advance. The decay
steepens as the qubit number grows, leaving extrapolations to
larger devices unsupported. A better optimizer exists, but it
gains a few percent and the peakedness it reaches decays at
nearly the same rate. Aaronson
and Zhang attribute the difficulty to a barren plateau. The
plateau is present, but the second-order amplitude data are
depth-independent while the attained peakedness is not.
Independent runs land on uncorrelated solutions, yet paths
connecting them stay well above the scale of a random state, so
fragmentation fails at that scale. In the deep limit we take the
scrambled state Haar-random. There we prove that no search over
a polynomial-parameter Lipschitz family, exhaustive search
included, beats the scale of a random state by more than a
$\mathrm{poly}(n)$ factor on average. Both
branches of the Aaronson--Zhang alternative therefore stay open. A
generation method must beat a preregistered baseline, and a
hardness argument must accommodate a solution set connected at the
scale of a random state.
\end{abstract}
 
\newpage
\setcounter{page}{1}
\hypersetup{pageanchor=true}

% ---------------------------------------------------------------- body ----
\section{Introduction}
\label{sec:intro}

Random-circuit sampling is the leading experimental route to quantum
advantage, and it carries a defect. Verifying the output costs
classical time exponential in the qubit number $n$, so an experiment
large enough to defeat spoofing is too large to check. Aaronson and
Zhang proposed peaked circuits as a way around this
\cite{aaronson2024peaked}. A peaked circuit looks random but
concentrates far more weight than chance on a
single computational-basis string. A verifier who knows that string
checks the device in a handful of shots. The verification cost does
not vanish but moves to generation: someone must produce circuits
that scramble and peak at once.

\begin{figure}[H]
\centering
\includegraphics[width=\textwidth]{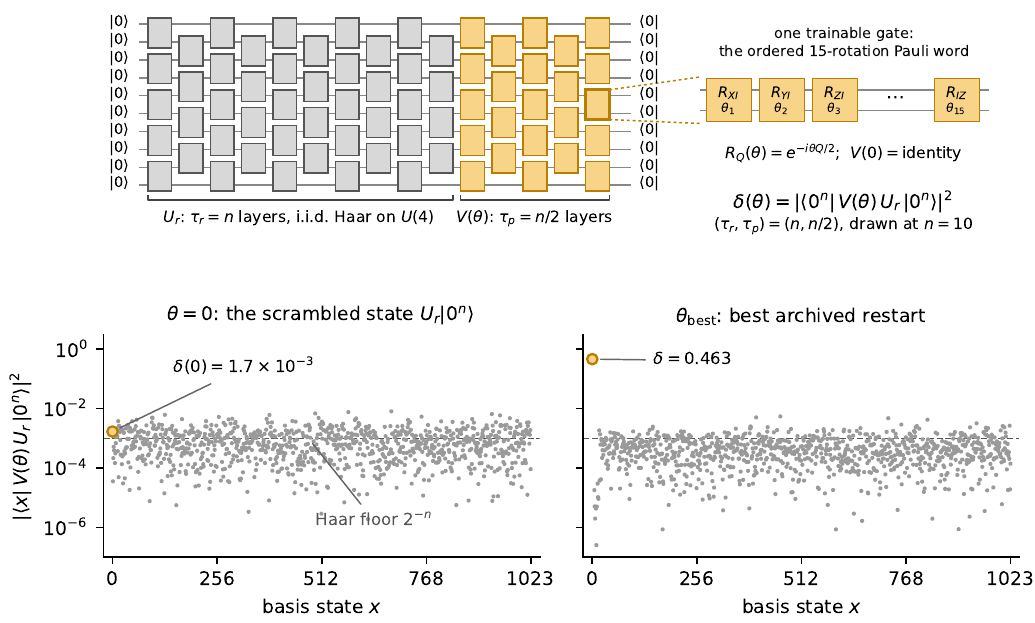}
\caption{The generation problem. \emph{Top:} the geometry of
Sec.~\ref{sec:setup}. \emph{Bottom:} the exact output distribution
of a single instance at $n = 10$. At $\theta = 0$ (left) the
scrambled state is a Porter--Thomas cloud. At the best of $32$
restarts (right) the same instance concentrates $\delta = 0.463$
on $0^n$, $474$ times the $2^{-n}$ floor (dashed), while the
remainder stays Porter--Thomas-like. Fifteen of the $1024$ strings
fall below the axis range.}
\label{fig:intro}
\end{figure}
 
Aaronson and Zhang generate peaked circuits variationally. A
brickwall of Haar-random gates $U_r$ of depth $\tau_r$ is followed by
a trainable brickwall $V(\theta)$ of depth $\tau_p$. The parameters
$\theta$ are optimized by gradient descent to maximize the peakedness
$\delta(\theta) = |\bra{0^n} V(\theta)\, U_r \ket{0^n}|^2$
(Fig.~\ref{fig:intro}). At $(\tau_r, \tau_p) = (n, n/2)$, the
peakedness they reach decays as $1.189^{-n}$, against $2^{-n}$ for a
Haar-random state \cite{aaronson2024peaked}. If that base persists,
$n = 50$ leaves $\delta \approx 5\times 10^{-4}$, accessible to
experiment. They state the alternative: either their gradient
descent stalled in local optima and a better optimizer unlocks
higher peakedness, or no efficient generation method exists. In the
second case the construction is of limited use for advantage
experiments.

\newpage
Aaronson and Zhang attribute the difficulty to a barren plateau
(gradient variance vanishing exponentially in $n$ at initialization)
\cite{mcclean2018barren}. Their Theorem~2.4 gives a depth bound:
$1/\mathrm{poly}(n)$ peakedness at random depth $\tau_r \gg n$
requires $\tau_p = \Omega\bigl((\tau_r/n)^{0.19}\bigr)$. Their
Theorem~A.1 constructs, for two random layers, an explicit single
peaking layer.

The neighboring results leave the search over $\theta$ untouched.
Estimating the peakedness of a random peaked circuit to additive
error $2^{-\mathrm{poly}(n)}$ is average-case
\#P-hard \cite{zhang2025complexity}. Deciding peakedness on a
diagonal instance (the same unknown string at input and output) is
QCMA-complete \cite{heuristic2025peaked}. Both place the problem in
complexity classes and say nothing about the search. Hardware
demonstrations and classical-simulation countermeasures are under
way \cite{ibm2026tracker,ferreira2026sparse}, and none of that
activity maps the search.

The alternative turns on none of these results. Neither the depth
bound nor the explicit construction says whether an optimizer finds
$\theta$. What separates the two branches is the optimization
landscape of $\delta(\theta)$ over the variational parameters, and
no existing result maps it. Aaronson and Zhang optimize and
extrapolate the law fitted to the output; the landscape that any
such law presupposes stays unmeasured. As long as it is unmeasured,
a stalled optimizer is indistinguishable from an absent method at
accessible sizes. The extrapolation to $n = 50$ rests on a fitted
base with no structural support. The stake is the proposal itself:
the verifier must build the circuit, so peaked-circuit advantage
stands or falls with generation.

We approach the alternative through three measurements: the reach
of optimizers, the exact statics of the landscape, and the geometry
of the solution set. The reach is the peakedness an optimizer
attains at a given budget. Three named obstructions could have
carried the difficulty, and each is confronted with the
corresponding measurement. The barren plateau is read against the
initialization statistics \cite{mcclean2018barren}. The overlap-gap
property derives hardness from clustering of near-optimal solutions
and is confronted with direct connectivity scans
\cite{gamarnik2021overlap}. Recorded trajectories confront the
trap-based obstructions of variational landscapes
\cite{anschuetz2022traps}.

In this work, we measure the landscape on which the alternative
turns. The protocol's falsifiers and their costs were fixed before
the runs existed, and we call such items registered. We call the $2^{-n}$ peakedness of a Haar-random state
the trap scale. Our contributions are the following.
\begin{enumerate}
\item We compute the exact statics of the landscape as a random
field: the finite-depth covariance kernel and the moments of
$\delta$ to fourth order (Sec.~\ref{sec:statics}). The mean is
exactly flat at every depth, so all structure resides in the
fluctuations.
\item We prove the deep-limit ceiling, taking the scrambled state
Haar-random (Sec.~\ref{sec:ceiling}). No search over a
polynomial-parameter Lipschitz family, exhaustive search included,
beats the trap scale by more than a $\mathrm{poly}(n)$ factor on
average.
\item We bound the finite-depth ceiling $\delta^*$ (the supremum of
$\delta$ at fixed peaking depth) by an entanglement budget
(Sec.~\ref{sec:deltastar}). The first-rung theorem gives the exact
maximum at one peaking layer, and its prediction is
confirmed at $10^{-6}$ on five of six points.
\item We measure the reach law at a frozen budget, fixed in
advance, and at convergence: the decay steepens with $n$
(Sec.~\ref{sec:steepening}). A fixed-base law fitted to $n \le 14$
is rejected at $p = 0.026$ frozen and $0.041$ converged, leaving the
extrapolation of Aaronson and Zhang to $n = 50$ unsupported.
\item We show a better optimizer exists: L-BFGS-B ends above
converged Adam at $n = 16$, at reach ratio $1.0389$
(Sec.~\ref{sec:robustness}). The peakedness it reaches still falls,
at local base $1.31$ on the last interval against Adam's $1.33$.
\item We read the barren plateau against the reach law
(Sec.~\ref{sec:bp}). The plateau is present at initialization, but
the exact second-order amplitude data are depth-independent while
the reach is not.
\item We map the geometry of the solution set: it does not fragment
at the trap scale (Sec.~\ref{sec:solgeometry}). Reachable solutions
are pairwise decorrelated, yet the re-optimized paths connecting
them stay $10^2$--$10^3$ times above that scale.
\item We synthesize the picture as shrinking reach without a
trap-scale overlap gap and state three conjectures, each with a
registered falsifier (Sec.~\ref{sec:hardness}). Our own measurement
triggers the falsifier of the hardness conjecture, and the
conjecture is withdrawn.
\end{enumerate}
Neither branch of the alternative closes: generation must now beat a
registered baseline, and hardness must accommodate a solution set
connected at the trap scale.

Throughout, circuits are noiseless and the ensemble is a single
brickwall family under a fixed gate parametrization. The operating
regime is $(\tau_r, \tau_p) = (n, n/2)$ unless stated. The grid is
$18$ instances per size at $n = 8$--$16$, at the frozen and at the
converged budget. A reach value is tied to a named optimizer at a
named budget. It measures the pair of landscape and algorithm, not
the landscape alone. Every reach claim is quantified over the
stable-local class of Sec.~\ref{sec:protocol}: gradient-flow-like
dynamics with $\mathrm{poly}(n)$-bounded step sizes,
$\mathrm{poly}(n)$ steps, and independent restarts.

The body develops the contributions in the order listed;
Sec.~\ref{sec:related} relates the results to prior work,
Sec.~\ref{sec:limits} collects the limitations,
Sec.~\ref{sec:discussion} concludes, and
Appendices~\ref{app:proofs}--\ref{app:cuberoot} hold the proofs, the
methods, the numerical controls, the registration record, and the
cube-root test.
 \section{Setup, conventions, and protocol}
\label{sec:setup}

\subsection{The field}
\label{sec:field}

Fix $n$ even, $n \ge 4$. The instance $U_r$ is a one-dimensional
open-boundary
brickwall circuit of depth $\tau_r$ (the random depth) whose
two-qubit gates are i.i.d.\ Haar on $U(4)$; layer $0$ covers all $n$
qubits by disjoint pairs. The variational circuit $V(\theta)$
continues the same brick pattern for $\tau_p$ layers (the peaking
depth); each gate is the ordered product of fifteen Pauli rotations
$\prod_a e^{-i\theta_{p,a} Q_a/2}$ over the fixed two-qubit Pauli
word sequence, giving the identity at $\theta = 0$, per-coordinate
$2\pi$-periodicity up to sign, and
$P = 15 \times (\text{gates})$ real parameters, $\Theta(n^2)$ at
the operating $\tau_p = n/2$. The
object of study is the random field
\begin{equation}
\delta(\theta) \;=\;
\big|\bra{0^n}\, V(\theta)\, U_r\, \ket{0^n}\big|^2
\;=\; |\braket{w(\theta)}{\varphi}|^2,
\label{eq:field}
\end{equation}
with $\ket{\varphi} = U_r\ket{0^n}$ the scrambled state,
$\ket{w(\theta)} = V(\theta)^\dagger \ket{0^n}$ the probe state, and
randomness carried by the instance; $\delta(\theta)$ is the
peakedness. Fixing the peak string to $0^n$ costs no generality: the
generation problem of Ref.~\cite{aaronson2024peaked} allows any
target $s$, and conjugating by $\bigotimes_j X_j^{s_j}$ maps $s$ to
$0^n$, leaving the Haar-random section Haar-random; the conjugation is
absorbed into the boundary layer of $V$. The variational instance
Ref.~\cite{aaronson2024peaked} optimizes numerically (their
Eq.~(9)) is $\max_\theta \delta(\theta)$ at
$(\tau_r, \tau_p) = (n, n/2)$; their generation problem itself
(their Open Problem~1.4) admits any efficient algorithm, and the
statements here quantify over this manifold only. The
instance-independent probe overlap
$F(\theta, \theta') = |\braket{w(\theta)}{w(\theta')}|^2$ is the
correlation coordinate of the field; the solution overlap
$q_{ij} = |\braket{\Psi_i}{\Psi_j}|^2$ for
$\ket{\Psi_i} = V(\theta_i) U_r \ket{0^n}$ is the ensemble
observable of the experiments and the coordinate in which solutions
are compared, equivalent to the Fubini--Study distance up to
monotone reparametrization. Table~\ref{tab:notation} collects the
principal notation.

\begin{table}[H]
\caption{Principal notation.}
\label{tab:notation}
\begin{tabularx}{\textwidth}{@{}lX@{}}
\toprule
$F(\theta, \theta')$ & the instance-independent probe overlap
  $|\braket{w(\theta)}{w(\theta')}|^2$, the correlation coordinate of
  the field (Sec.~\ref{sec:field}) \\
\midrule
$q_{ij}$ & the solution overlap $|\braket{\Psi_i}{\Psi_j}|^2$ for
  $\ket{\Psi_i} = V(\theta_i) U_r \ket{0^n}$, the ensemble observable
  of the experiments (Sec.~\ref{sec:field}) \\
\midrule
$\hat q$ & the floor-normalized pair overlap,
  $(q - \delta_i\delta_j)/(1 - \delta_i\delta_j)$
  (Sec.~\ref{sec:protocol}) \\
\midrule
trap scale & the $2^{-n}$ Haar floor of the field
  (Sec.~\ref{sec:protocol}) \\
\midrule
stable-local class & gradient-flow-like dynamics with
  $\mathrm{poly}(n)$-bounded step sizes, $\mathrm{poly}(n)$ steps,
  and independent restarts; the class as registered, containing
  Adam and plain SGD (Sec.~\ref{sec:protocol}) \\
\midrule
amended stable-local class & the registered class restricted by
  the pacing criterion
  $\lVert\Delta\theta_t\rVert = \Theta(\eta_t)$, adopted post hoc
  and used for no registered verdict (Sec.~\ref{sec:protocol}) \\
\midrule
$\delta^*$ & the supremum
  $\delta^*(\tau_p; U_r) = \sup_\theta \delta(\theta)$, a sharp
  instance quantity
  (Sec.~\ref{sec:deltastar}) \\
\midrule
$m_k$, $r_k$ & the normalized moments
  $\Ee[\delta^k]/(k!\,D^{-k})$, $D = 2^n$; the pole-free form
  $m_k / m_2^{\binom{k}{2}}$, equal to $1$ identically under pair
  dominance (Sec.~\ref{sec:moments}, Eq.~\eqref{eq:pairratio}) \\
\midrule
$\Lambda_{\rm ent}$, $\Lambda_{\rm MPS}$ & the top-$\chi_j$ Schmidt
  mass of the
  scrambled state; refined by the maximum
  $\Lambda_{\rm MPS} \le \Lambda_{\rm ent}$
  over matrix-product states of bond dimensions $\chi_j$
  (Eq.~\eqref{eq:truncation}) \\
\midrule
$R(n, B)$ & the expected best-of-$B$ peakedness under the frozen
  protocol (Sec.~\ref{sec:dynamics}) \\
\midrule
$\rho_{\rm arm}(n)$ & the reach ratio $R_{\rm arm}(n)/R_{\rm conv}(n)$
  of an optimizer arm to the converged cell
  (Sec.~\ref{sec:robustness}) \\
\midrule
$\varrho(n)$ & the corrugation floor, the median string-path minimum
  relative to the endpoint level (Sec.~\ref{sec:shelf}) \\
\midrule
$b_{\rm triv}$ & the peakedness of the instance truncated to its
  first $\tau_r - \tau_p$ random layers
  (Appendix~\ref{app:cuberoot}) \\
\bottomrule
\end{tabularx}
\end{table}
 
\subsection{Symmetries of the parametrization and of the objective}
\label{sec:gauge}

Three redundancies make the Euclidean geometry of $\theta$
physically meaningless. First, gate phases are confined to the
discrete center $\mathbb{Z}_4$ (each gate has unit determinant; a
$2\pi$ coordinate shift flips its sign), so the landscape lives on
the torus $(\mathbb{R}/2\pi\mathbb{Z})^P$. Second, the gate map
$\mathbb{R}^{15} \to SU(4)$ has generically discrete fibers (lattice
translations plus finitely many sheets), with surjectivity for the
fixed word order tested numerically. Third, the parametrization
carries a continuous gauge whose dimension is a counting, not a
measurement. Call an \emph{internal wire segment} a qubit $q$
together with two consecutive gates of $V$ acting on it, no gate
between them touching $q$. Inserting $u u^\dagger$, $u \in SU(2)$,
on $q$ between the two gates of a segment and absorbing each factor
into its endpoint gate leaves $V(\theta)$ exactly invariant while
displacing $\theta$: the factors commute with every gate they cross,
and regularity of the word map lifts the dressed gates back to
angles. Sec.~\ref{sec:solgeometry} counts these directions and
measures the count at the solutions.

\subsection{Frozen protocol and conventions}
\label{sec:protocol}

Unless stated otherwise, ensembles use one frozen protocol: Adam
\cite{kingma2015adam} at learning rate $0.05$ (decayed
$\times 0.5$ every $300$ steps), $400$ steps with an early stop
once the per-step change of $\delta$ falls below $10^{-8}$ after
step $100$, and initialization
$\theta_0 \sim \mathcal{N}(0, \sigma^2 I_P)$ with $\sigma = 0.1$.
We define the \emph{stable-local class} as gradient-flow-like
dynamics with $\mathrm{poly}(n)$-bounded step sizes,
$\mathrm{poly}(n)$ steps, and independent restarts. This is the
class as registered, and every reach claim below is quantified
over it: the protocol's Adam is in it, and so is plain SGD. Kept
strictly apart, the \emph{amended stable-local class} is the
registered class restricted by a pacing criterion: per-step
parameter displacement paced by the step-size schedule $\eta_t$,
$\lVert\Delta\theta_t\rVert = \Theta(\eta_t)$ along the run, as
preconditioned or normalized first-order methods provide. The
criterion is outcome-independent in form, reading displacements
along the trajectory and never the reached peakedness, but not in
provenance: adopted after the SGD run of Sec.~\ref{sec:robustness},
it grades no registered prediction here. An unqualified
``stable-local class'' below always denotes the registered class.
For an instance, tolerance $\varepsilon$, and restart budget $B$,
the \emph{reachable
ensemble} is the set of final points of $B$ independent restarts,
filtered to
$\delta_{\rm final} \ge (1-\varepsilon)\,\delta_{\rm best}$ with
$\delta_{\rm best}$ the batch best; a restart below this filter is
stalled. The floor-normalized pair overlap is
$\hat q = (q - \delta_i\delta_j)/(1 - \delta_i\delta_j)$, with
$\delta_i, \delta_j$ the two solutions' final peakedness values; it
removes the overlap that two solutions sharing nothing but the peak
string already have. Pairs with $\hat q \in [0.25, 0.75]$ sit in
the \emph{intermediate band}. We define the \emph{trap scale} as the $2^{-n}$
Haar floor of the field; a connectivity path whose level falls to it
carries no more signal than a fresh restart. \emph{Registered} marks
what was fixed before the runs that test it. The corrugation
protocol, statistic, and verdict rule sit in the archived
\texttt{REGISTRATION.md}, and the optimizer-class arms, statistic,
and outcome costs in the archived \texttt{REGISTRATION-CONVERGED.md}.
The moment bands of Sec.~\ref{sec:moments} and the four predictions
of Appendix~\ref{app:registration} sit in the working record and
enter the archive as the author's declaration.

Uncertainty conventions: $x \pm y$ denotes one standard error (of a
fit parameter or of an instance mean, as stated); parenthesized
digits denote the spread over instances in the last digit; error
bars in figures are standard errors of the instance mean.
\FloatBarrier
 \section{Exact statics of the field}
\label{sec:statics}

\subsection{The mean is exactly flat}
\label{sec:flatmean}

For any $\tau_r \ge 1$ and every $\theta$,
$\Ee[\delta(\theta)] = 2^{-n}$: layer $0$ covers all qubits with
independent Haar pairs, the single-pair twirl sends any input to
$I_4/4$, and unitary invariance preserves $I/2^n$ thereafter. The
landscape has no mean slope anywhere; by dominated differentiation
the mean gradient and mean Hessian vanish identically at every
depth, so all structure resides in the fluctuations.

\subsection{Second-order amplitude data}
\label{sec:amplitude}

Write $D = 2^n$. The first-moment twirl gives, at every depth
$\tau_r \ge 1$ and for the amplitude
$\psi(\theta) = \braket{w(\theta)}{\varphi}$,
\begin{equation}
\Ee|\psi(\theta)|^2 = \frac{1}{D}, \qquad
\Ee\big[\psi(\theta)\,\overline{\psi(\theta')}\big] =
\frac{\braket{w(\theta)}{w(\theta')}}{D},
\label{eq:hypotheses}
\end{equation}
the increment identity
$\Ee|\psi(\theta) - \psi(\theta')|^2 = \|w(\theta) -
w(\theta')\|^2/D$ following from the two, and phase-invariance of
every gate ensemble makes $\Ee[\psi]$ and the pseudo-covariance
vanish exactly. The true field's mean square, covariance, and
increment metric are thus exactly Haar-like at every depth.

\subsection{The exact finite-depth covariance kernel}
\label{sec:kernel}

With $A(\theta) = \ket{w(\theta)}\bra{w(\theta)}$, and primes
denoting evaluation at $\theta'$, the second moment is a two-copy
average,
\begin{equation*}
\Ee[\delta\delta'] = \Tr[(A \otimes A')\,
\Ee[(\ket{\varphi}\bra{\varphi})^{\otimes 2}]].
\end{equation*}
Each Haar gate twirls the two-copy space onto
$\mathrm{span}\{\mathrm{id}, \mathrm{swap}\}$, so per-site labels
$\pi_i \in \{e, s\}$ propagate through the brickwall with exact
local rules derived from the $2\times2$ Weingarten system at local
dimension $4$ (Corollary~2.4 of
Ref.~\cite{collins2006integration}; the brickwall transfer follows
Sec.~2, Eqs.~(16)--(18) of Ref.~\cite{hunterjones2019unitary}). A
first-layer gate emits its pair aligned, $e$ or $s$, with weight
$1/20$ each; aligned inputs pass unchanged; misaligned inputs
resolve to $ee$ or $ss$ with weight $2/5$ each. The raw rule is
positive but sub-stochastic (a misaligned input emits total weight
$4/5$); the exact invariants are the two operator pairings, and the
rule becomes an exactly stochastic Markov chain after dressing by
$4^{\#e} 2^{\#s}$. Contracting the final configuration $c$, with
$s$-region $S(c)$, against the probes gives, exactly at any depth,
\begin{equation}
\Ee[\delta(\theta)\,\delta(\theta')] =
\sum_{c\,\in\,\{e,s\}^n} W_{\tau_r}(c)\;
\Tr\!\big[\rho_{S(c)}(\theta)\,\rho_{S(c)}(\theta')\big],
\label{eq:kernel}
\end{equation}
with $W_{\tau_r}(c)$ the weight the two-copy transfer assigns to
configuration $c$ and the boundary the \emph{cross-purity} of the
two probe states on $S$. Eq.~\eqref{eq:kernel} is validated at $n = 8$ against
fresh Monte Carlo estimates at $25/25$ points across
$\tau_r \in \{1, 2, 4, 8, 16\}$, at a tolerance of two bootstrap
standard errors plus an absolute $10^{-3}$, the tightest point
passing at $2.2$ standard errors; and against $240$ further points
at $n \in \{8, 10\}$ across $\tau_r \in \{1, 2, 4, n, 2n\}$ and
three direction kinds, whose residuals in the correlation are
$-0.004 \pm 0.0035$, with no bias for any direction kind. Mixed
configurations damp at least geometrically ($4/5$ per domain wall
per layer), and the deep limit concentrates on $S = \emptyset$ and
$S = \mathrm{all}$, recovering the 2-design kernel
$\Ee[\delta\delta'] = (1 + F)/(D(D{+}1))$. In this limit the
field is the Porter--Thomas field on the variational manifold
\cite{porter1956fluctuations,boixo2018characterizing}, its full
correlation structure summarized by $F(\theta,\theta')$. At finite depth
the kernel depends on the probes through every
reduced cross-purity; a formula in $(F, n)$ alone would be wrong.

Differentiating Eq.~\eqref{eq:kernel} gives the exact gradient
covariance at initialization with the same weights and a
differentiated boundary; in the deep limit it reduces to
$\mathrm{Cov}[\partial_a\delta, \partial_b\delta] =
2 g_{ab}/(D(D{+}1))$ with $g$ the Fubini--Study metric of the probe
family.

\subsection{Third and fourth moments: a hierarchy below pair
dominance}
\label{sec:moments}

The same construction lifts to $k$ copies and computes the moments
$\Ee[\delta^k]$. Per-site labels now live in $S_k$: Gram matrix
$G_{\pi\pi'} = 4^{\#\mathrm{cyc}(\pi\pi'^{-1})}$, uniform
initialization $1/Z_k$ with $Z_k = 4\cdot5\cdots(4{+}k{-}1)$
($1/120$ at $k=3$, $1/840$ at $k=4$), transfer coefficients
$G^{-1} t$, where
$t_{\pi'} = 2^{\#\mathrm{cyc}(\alpha\pi'^{-1})}\,
2^{\#\mathrm{cyc}(\beta\pi'^{-1})}$ is the overlap of the gate's
misaligned input labels $(\alpha, \beta)$ with the permutation
$\pi'$, and boundary
$X_c = \bra{w}^{\otimes k} (\bigotimes_i \Pi_{\pi_i})
\ket{w}^{\otimes k}$, with $\Pi_\pi$ the copy-permutation operator
and $|X_c| \le 1$. For $k \ge 3$ the transfer is a \emph{signed}
measure (negative coefficients appear beyond the transposition
class), so evaluation is by exact summation with truncation
certified by the dropped $\sum |W|$; positivity is restored only at
the total moment. The permutation operators remain linearly
independent at $k = 4$ and local dimension $d = 4$; the Gram matrix first
degenerates at $k = 5$, where the transfer also carries $120$ labels
per site, and the ladder stops at the fourth moment. Self-tests cover
idempotence, trace preservation, the
Haar floors $m_k^{\rm Haar} = D^{k-1}/\prod_{j=1}^{k-1}(D{+}j)$, and
an independent Monte Carlo pipeline.

These moments are functions of the evaluation point $\theta$, not
numbers attached to the field. The values below are computed at the
protocol's initialization,
$\theta \sim \mathcal{N}(0, \sigma^2 I_P)$ with $\sigma = 0.1$: a
near-identity point, where the probe is close to $\ket{0^n}$ and
every reduced cross-purity is close to $1$; among the probe points
of Sec.~\ref{sec:probedep} only the product probe $\theta = 0$ shows
a larger excess. That initialization is the point the initialization statistics
of Sec.~\ref{sec:bp} sample.

\begin{figure}[H]
\centering
\includegraphics[width=0.55\textwidth]{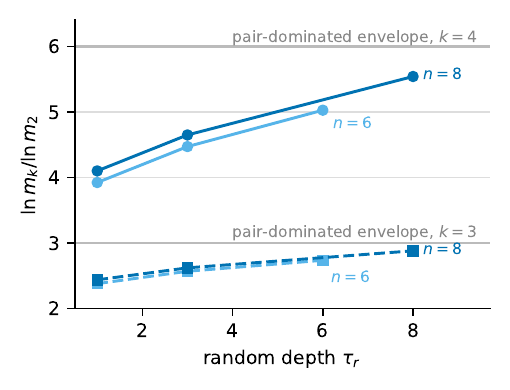}
\caption{Moment-hierarchy exponents $\ln m_k/\ln m_2$ versus random
depth $\tau_r$, for $k = 3, 4$ at $n = 6, 8$ (normalized moments
$m_k = \Ee[\delta^k]/(k!\,D^{-k})$; exact transfers). Every plotted
point carries either a complete configuration sum, at
$\tau_r \in \{1, 3\}$
for both sizes, or a truncation certificate: $0.33\%$ of $m_4$ at
$(n, \tau_r) = (6, 6)$ and $0.98\%$ at $(8, 8)$. The points $(6, 2)$ and $(8, 2)$
are excluded by their $15.9\%$ and $125\%$ certificates, and
$(6, 12)$ and $(8, 16)$ because $m_2$
sits within $1.1\%$ of the exponent's pole at $m_2 = 1$; the figure
script prints each exclusion.
The hierarchy sits
below the pair-dominated log-envelope $m_k = m_2^{\binom{k}{2}}$
everywhere measured, whose exponents are $3$ and $6$ (see text),
with the gap narrowing over the measured depths.}
\label{fig:moments}
\end{figure}
 
Fig.~\ref{fig:moments} shows the resulting hierarchy in terms of
the normalized moments $m_k = \Ee[\delta^k]/(k!\,D^{-k})$. At the
operating depth $\tau_r = n$ and at that probe: $m_2 = 1.118$,
$m_3 = 1.357$, $m_4 = 1.752$ at $n = 6$ and $m_2 = 1.164$,
$m_3 = 1.549$, $m_4 = 2.320$ at $n = 8$, with $k$-exponents
$\ln m_4/\ln m_2 = 5.03$ and $5.53$, inside the band $[4.5, 6)$
registered before either $\tau_r = n$ value existed; at $n = 8$ the
exponent lands in the upper half of the band, the half the
registration favored.

We assume, as the reference model against which the measured moments
are read, $\delta = D^{-1}\mathcal{E}\,Y$ with $\mathcal{E}$ a
standard exponential, the local Porter--Thomas factor, and $Y > 0$ a
local modulus independent of $\mathcal{E}$; the independence is the
modeling step, and the deep limit realizes it up to $O(1/D)$, with
$Y \equiv 1$ in the limit. Under it, $\Ee[\delta^k] = k!\,D^{-k}\,\Ee[Y^k]$, so
normalizing by $k!\,D^{-k}$ isolates the modulus
exactly, $m_k = \Ee[Y^k]$. The flat mean of
Sec.~\ref{sec:flatmean} fixes $\Ee[\delta] = D^{-1}$ and hence
$\Ee[Y] = 1$, and for a log-normal modulus of unit mean the higher
moments close on $m_2$ alone, that is
\begin{equation}
m_k \;=\; m_2^{\binom{k}{2}}
\label{eq:envelope}
\end{equation}
identically, with $\ln m_k/\ln m_2 = \binom{k}{2}$: $3$ at $k = 3$
and $6$ at $k = 4$, independently of the size of the excess. We call
Eq.~\eqref{eq:envelope} the pair-dominated log-envelope. A Gaussian
field in the amplitude is not the relevant reference: it makes
$\delta$ exponential, so $Y \equiv 1$, $m_k = 1$ for every $k$, and
the exponent is undefined. A companion registration on the third
moment at $(n, \tau_r) = (8, 8)$ missed its numeric band.
Appendix~\ref{app:registration} records the bands, the verdict,
and the units error behind the miss.

The exponent $\ln m_k/\ln m_2$ is convenient at the operating depth
but singular in general: $m_2$ exceeds unity at $\tau_r = n$ and
falls below it in the deep limit, where the Haar floor gives
$m_2 = D/(D{+}1)$, so the ratio has a pole wherever $m_2$ crosses
$1$ and cannot be continued from the measured depths to the deep
limit. The pole-free form of the same comparison is
\begin{equation}
r_k \;=\; m_k \big/ m_2^{\binom{k}{2}},
\label{eq:pairratio}
\end{equation}
which equals $1$ identically under pair dominance. We measure, at
the committed probe draw at $\sigma = 0.1$, $r_3 = 0.971$,
$r_4 = 0.897$ at $n = 6$ and $r_3 = 0.982$, $r_4 = 0.932$ at
$n = 8$ (three independent probe draws at $n = 6$, $\tau_r = n$,
give $m_2 = 1.118$, $1.113$ and $1.101$, a $1.5\%$ range): the
fourth moment sits $7$--$10\%$ below the log-normal prediction and
the third $2$--$3\%$, so the modulus is sub-log-normal. The
reference needs no finite-size correction, since the Haar floor
itself satisfies $r_4 = 1.0009$ at $n = 6$ and $1.00006$ at
$n = 8$: the deep limit obeys pair dominance to parts in $10^{3}$
and $10^{5}$, and the measured departure is a hundred times larger
at $n = 6$ and a thousand times at $n = 8$.

\subsection{Probe dependence of the moment excess}
\label{sec:probedep}

At $(n, \tau_r) = (8, 8)$ the exact second moment falls as the probe
is dispersed, tracking the probe's own half-chain purity:
$m_2 = 1.168$ at $\theta = 0$ (the product state, where every
reduced cross-purity is exactly $1$), then $1.164$, $1.153$,
$1.080$ and $1.062$ at $\sigma = 0.1$, $0.2$, $0.5$ and $1$,
against half-chain purities $1.00$, $0.95$, $0.84$, $0.66$ and
$0.53$; the scan rescales one committed probe direction radially, so
its scales are correlated rather than independent draws. At the
optimizer's own solutions, the points the reach measurements of
Sec.~\ref{sec:dynamics} actually visit, $m_2 = 1.036$ at half-chain
purity $0.23$ (the five best restarts of instance $0$ at $n = 8$,
coinciding at the atom of Sec.~\ref{sec:atom},
$m_2 = 1.0358$--$1.0359$). At $n = 6$ the
same scan gives
$1.1233 \to 1.1180 \to 1.0982 \to 1.0396 \to 1.0401$ with purities
$1.00 \to 0.96 \to 0.82 \to 0.54 \to 0.59$, and $m_2 = 1.021$ at
the solutions (purity $0.51$): the last scale turns the purity back
up and $m_2$ follows it, so $m_2$ tracks the purity, not the scale.

The excess over the Gaussian value is therefore $16.4\%$ at the
initialization ($16.8\%$ at the product probe) and $3.6\%$ at the
reachable solutions, and the sub-log-normal deficit moves with it.
Exact $S_3$ and $S_4$ transfers at the two ends of the
near-identity range give, at $n = 6$, $r_3 = 0.9697$ and
$r_4 = 0.8915$ at the product probe against $0.9715$ and $0.8973$ at
$\sigma = 0.1$; at $n = 8$ the truncation certificate is $\pm 0.02$
on $r_4$, so both probe points read $r_4 = 0.93 \pm 0.02$ and carry
no trend there, while $r_3$ moves from $0.9810$ to $0.9817$. Three
deficits resolve beyond their certificates: $r_3$ at both sizes, and
$r_4$ at $n = 6$, where the shift is twice its certificate. All
three shrink as the probe disperses, over a range too narrow to say
more. The third moment can be followed further, and there the effect
is large: the committed $S_3$ scan ($n = 6$, $\tau_r = 6$; an
independent probe draw, the two draws differing by $0.13\%$ on
$r_3$) gives $r_3 = 0.973$ at $\sigma = 0.1$ against $0.998$ at
$\sigma = 0.5$: the deficit shrinks with dispersion, on one draw per
scale, and what remains at $\sigma = 0.5$ is the size of that
draw-to-draw spread, so the shrinkage factor is not resolved. The fourth moment at
a solution is computed: at the same five coinciding solutions of
instance $0$ at $n = 8$ (converged ensemble,
Sec.~\ref{sec:dynamics}), the exact transfer gives $m_3 = 1.110$,
resolved at a $3\times 10^{-6}$ certificate, and
$m_4 = 1.23 \pm 0.05$, the certificate now $4.4\%$ of the value
where the initialization-point one was $\pm 0.02$ on a larger
excess; in envelope units $r_3 = 0.9985$ and
$r_4 = 0.994 \pm 0.044$. At the points the reach measurements visit,
the third-moment deficit is small and resolved and the
fourth-moment one is consistent with zero within its certificate:
the sub-log-normal deficit closes along with the excess it rides on
(\texttt{analysis/moment\_probe\_converged.log}; one instance at one
size, the atom making its five best restarts a single point). The
falsification of the Gaussian model in
Sec.~\ref{sec:bp} survives untouched, since it needs only one
point with $m_2 \neq 1$ and every point qualifies.
Sec.~\ref{sec:limits} collects what the scan does not establish.
 \section{The deep-limit ceiling: no reach beyond
\texorpdfstring{$\mathrm{poly}(n)\,2^{-n}$}{poly(n) 2\^{}-n}}
\label{sec:ceiling}

The probe family is $1$-Lipschitz in the $\ell_1$ metric,
$\|w(\theta) - w(\theta')\| \le \|\theta - \theta'\|_1$
(per-coordinate generators are halved Pauli words), and the
landscape lives on the torus $[-\pi, \pi]^P$.

\begin{proposition}[Deep-limit ceiling]
\label{prop:ceiling}
Let $\ket{\varphi}$ be Haar on the unit sphere of $\mathbb{C}^D$,
$D = 2^n$, $n \ge 2$. Let $\theta \mapsto \ket{w(\theta)}$ be any
probe family on $[-\pi, \pi]^P$, $P \ge 2$, that is $1$-Lipschitz
in the $\ell_1$ metric, and write
$\delta(\theta) = |\braket{w(\theta)}{\varphi}|^2$. There is a
universal constant $C$ (one may take $C = 7$) such that
$\Ee_\varphi\big[\max_\theta \delta(\theta)\big] \le
C\, P\,(n + \ln P)\, 2^{-n}$.
\end{proposition}

The proof is a net argument: the pointwise tail is exactly
$\mathrm{Beta}(1, D{-}1)$, an $\ell_1$-net of radius $2^{-n/2}$ has
$\ln(\text{cardinality}) \le \tfrac{3}{2} P (n + \ln P)$, the
off-net transfer costs a factor $2$ plus $2^{1-n}$, and the union
bound integrates against the exponential tail
(Appendix~\ref{app:proofs}). No Gaussianity is assumed; the only
idealization is taking $\ket{\varphi}$ Haar-random, the search
space remaining the $P$-parameter manifold. In the deep limit,
generation beyond $\mathrm{poly}(n)\, 2^{-n}$ peakedness is
impossible in expectation over the instance, for any algorithm
searching that manifold, exhaustive search included. The bound
says nothing about peaking unitaries of exponential size, where the
unitary carrying $\ket{\varphi}$ to $\ket{0^n}$ attains
$\delta = 1$.

In the $C = 7$ form of
Proposition~\ref{prop:ceiling} the bound is not numerically binding
on the accessible grid: it reads $2.19$ at $n = 16$ and first falls
below unity at $n = 18$ ($0.77$). The constant is not optimized:
evaluated with the exact net cardinality and the Lipschitz constant
$1/2$ that the generators give, the same chain crosses unity at
$n = 14$, inside the measured grid. The separation from the
measured reach (Sec.~\ref{sec:dynamics}) remains one of scaling
form, not of magnitude.

Eq.~\eqref{eq:hypotheses} puts the true field's mean square at
$1/D$ and bounds its increments. Lemma~\ref{lem:gaussian} caps the
expected maximum of every Gaussian field carrying those data.

\begin{lemma}[Gaussian comparison]
\label{lem:gaussian}
Let $\psi(\theta)$, $\theta \in [-\pi,\pi]^P$, be any centered
separable complex Gaussian field with
$\Ee|\psi|^2 \le C_1 2^{-n}$ and increment metric bounded by
$C_2\, 2^{-n/2}\, \|\theta - \theta'\|_1$. Then
$\Ee\big[\sup_\theta |\psi(\theta)|^2\big] \le
C_0\, C_1\, P \ln\!\big(2 + C_2 P/\sqrt{C_1}\big)\, 2^{-n}$ with
$C_0$ universal and not made explicit here; for $C_1$ and $C_2$ of
order one, as they are for the true field, this is
$C_0\, P \ln(2{+}P)\, 2^{-n}$.
\end{lemma}

The proof lifts to a real field on the phase circle, then applies
Dudley's entropy bound and Gaussian concentration (Theorems~8.1.3
and~5.2.2 of Ref.~\cite{vershynin2018high}). Chaining removes the
factor $n$, and it needs sub-Gaussian increments, not Gaussianity.
In the deep limit the true field supplies them through the Beta
tail, so the $P\ln(2{+}P)$ order bounds the true field's maximum
there as well, with a constant not made explicit
(Appendix~\ref{app:proofs}). The gap to
Proposition~\ref{prop:ceiling}'s $P(n + \ln P)$ measures that
proposition's single-scale union bound; what the proposition keeps
is the explicit constant.
 \section{The finite-depth ceiling \texorpdfstring{$\delta^*$}{delta*}}
\label{sec:deltastar}

At finite $\tau_p$ the supremum $\delta^*(\tau_p; U_r) = \sup_\theta
\delta(\theta)$ is a sharp instance quantity (Sec.~\ref{sec:atom}), and
structural constraints of the variational manifold replace the
deep-limit ceiling of Sec.~\ref{sec:ceiling}.

\subsection{An entanglement-budget bound}

A two-qubit gate has operator Schmidt rank at most $4$ across any cut it
straddles, so a cut $j$
crossed by $\nu(j)$ peaking gates obeys $\mathrm{rank}_j\, w(\theta) \le
\chi_j = \min(4^{\nu(j)}, 2^{\min(j, n-j)})$, and by Eckart--Young
\cite{eckart1936approximation}
\begin{equation}
\delta^* \;\le\; \Lambda_{\rm ent} \;=\; \min_j\, \sum_{i \le \chi_j}
\lambda_i^{(j)}(\varphi),
\label{eq:truncation}
\end{equation}
the top-$\chi_j$ Schmidt mass of the scrambled state, refined by the
maximum $\Lambda_{\rm MPS} \le \Lambda_{\rm ent}$ over matrix-product
states of bond dimensions $\chi_j$
\cite{vidal2003efficient,oseledets2011tensor}; alternating least
squares approaches that maximum from below, so the quoted
$\Lambda_{\rm MPS}$ understates a valid upper bound. Since
$\nu(j) \approx
\tau_p/2$, the bound is vacuous at the operating $\tau_p = n/2$ and
bites below it, where the manifold binds first: at $\tau_p = n/4$
(measured at $n = 8$) the reached
value sits at $0.44$--$0.46$ of $\Lambda_{\rm ent}$, and on the
reference instance (instance $0$ at $n = 8$) the reached
$\delta(\tau_p = 2) \approx 0.36 \ll \Lambda_{\rm MPS} \approx 0.75 <
\Lambda_{\rm ent} \approx 0.78$.

\subsection{The atom and the first-rung theorem}
\label{sec:atom}

At shallow peaking depth the protocol's endpoint is an ``atom'': at
$(n, \tau_p) = (8, 2)$ the ten best of $60$ independent restarts agree
with the batch best to $\sim 10^{-7}$ (spreads $2.7\times10^{-8}$,
$1.2\times10^{-7}$, $3.0\times10^{-8}$ over the three instances),
across decorrelated maximizers; the full $60$-restart ensemble spans at
most $5.1\times10^{-3}$. An atom certifies the reproducibility of the
protocol's endpoint, not its optimality: at $(8, 1)$, instance $0$,
an atom of width $2.2\times10^{-8}$ sits $7.2\times10^{-4}$ below
the exact maximum of Theorem~\ref{thm:firstrung}. The endpoint is
highly degenerate and instance-sharp, and where ground truth exists
its bias is one-signed: a measurable target for the theory of
$\delta^*$, from below.

The atom is a property of $n = 8$ and does not survive size. The
relative spread of the ten best restarts, median over instances, reads
$5.9\times10^{-7}$ at $n = 8$, where $13$ of the $18$ instances are
atomic below $10^{-5}$, and then $1.8\times10^{-2}$,
$3.1\times10^{-2}$, $4.8\times10^{-2}$ and $6.0\times10^{-2}$ at
$n = 10$--$16$, the last on the first four instances
($3.7\times10^{-2}$ on all eighteen), where no instance is. So $\delta^*$ at the operating
regime is measured at $n = 8$ from below, by the atom; at the one
configuration with ground truth, $\tau_p = 1$, the protocol's
one-signed bias reaches $4.4\times10^{-3}$ relative at instance $0$
and stays below $10^{-6}$ at the six other points. Beyond $n = 8$, $\delta^*$ is
unmeasured, and this section's theory is anchored at one size. At the
first rung the theory closes:

\begin{theorem}[First rung]
\label{thm:firstrung}
For even $\tau_r$, $\delta^*(\tau_p = 1)$ is at most the maximal
squared overlap of the scrambled state with a state product across
the two-qubit blocks of the peaking layer's partition, the squared
entanglement eigenvalue $\Lambda_{\max}^2$ of
Ref.~\cite{wei2003geometric} for that partition, with equality if
the block orbit $\{V_p(\theta_p)^\dagger\ket{00}\}$ is dense in the
unit sphere of $\mathbb{C}^4$ up to phase.
\end{theorem}

\emph{Proof sketch.} At $\tau_p = 1$ the peaking circuit is a single
brick layer, a tensor product of independent two-qubit gates on the
$n/2$ disjoint brick pairs, so the variational image $\{w(\theta)\}$
lies in the block-product states and the bound is the maximum over that
set, attained by compactness. Density of the block orbit in the unit
sphere of $\mathbb{C}^4$ turns that inclusion into a density of sets
and the bound into an equality, a supremum over a dense set equaling
the maximum over its compact closure. The hypothesis is weaker than
surjectivity of the gate map onto $SU(4)$; it is tested numerically
for the fixed word order (Sec.~\ref{sec:gauge}) and not proved: a
modeling hypothesis of the equality, as Gaussianity is of
Lemma~\ref{lem:gaussian}. Appendix~\ref{app:proofs} gives the
details.

The prediction is confirmed at
$10^{-6}$ on five of six points: the block-ALS $\Lambda_{\max}^2$ of the
scrambled state ($\tau_r = n$, twelve ALS restarts per point,
converged in that count by Appendix~\ref{app:controls}) and the
best $\tau_p = 1$ restart agree to $10^{-6}$ at $n = 6$, instances
$0$--$2$ ($0.410272$ at instance $0$), and at $n = 8$, instances
$1$--$2$. At $n = 4$ a closed form from the top Schmidt vector
matches to $6\times10^{-8}$. At $n = 8$, instance $0$, the protocol
falls $7.2\times10^{-4}$ short. A later control resolves
that point exactly. Inverting the fifteen-rotation word block by
block reaches the ALS-optimal product state to $5\times10^{-16}$
per block, and $\delta$ at the assembled parameters equals
$\Lambda_{\max}^2$ to $10^{-15}$
(\texttt{analysis/firstrung\_inversion.log}). The equality branch
holds at all six points; what fails at the sixth is the frozen
protocol, whose ten best of $60$ restarts agree to
$2.2\times10^{-8}$ while ending $7.2\times10^{-4}$ below the
constructed maximum. The shortfall is one-signed at every point
measured: the protocol ends at or below the block-product value
$\Lambda_{\max}^2$, never above.
 \section{Reach dynamics}
\label{sec:dynamics}

The consolidated campaign measures the fixed protocol's reach:
$18$ instances per size at $n = 8$--$14$ and, through a registered
catch-up, at $n = 16$ ($200$ restarts each), fresh $800$-restart
ensembles at $n = 10, 12$, deep anchors at three instances
(Sec.~\ref{sec:bp}),
optimizer-robustness controls, and connectivity scans at matched
resolution (Sec.~\ref{sec:shelf}). Our published $n = 16$ values
rest on the first four
instances; per the registration, every $n = 16$ number below is
reported on both that set and all eighteen, whichever way the
change goes. $R(n, B)$ denotes the expected best-of-$B$ peakedness
under the frozen protocol.

The protocol stops at $400$ Adam steps or at the convergence
tolerance, whichever comes first; the cap binds unevenly across
sizes: at $n = 8$, $61\%$ of restarts reach it and only $3$ of
the $18$ instance-best restarts do, whereas at $n \ge 10$,
$94$--$99\%$ of restarts and \emph{every} instance-best restart
do. The reach reported at $n \ge 10$ is thus a truncated
optimizer's, not the protocol's at convergence; the cost is
measured in Appendix~\ref{app:controls} by quadrupling the cap on
the same restart seeds, and settled by the registered converged
campaign, which reran the full $18$-instance grid at $n = 8$--$16$
to optimizer convergence, acceptance criterion met at every size.
Convergence shifts the level of the reach curve and bends the
steepening statistic by ten percent; the fixed-base rejection
stands (Sec.~\ref{sec:steepening}).

\subsection{The reach law steepens}
\label{sec:steepening}

The fixed-base summary over the full grid (Fig.~\ref{fig:scaling})
is $1.219$ per qubit; its standard error is $0.0064$ under the
fixed-base model and $0.017$ once scaled by the misfit
$\sqrt{\chi^2/\mathrm{dof}}$. The scaled error is the one we
quote, since the fit rejects its own model
($\chi^2/\mathrm{dof} = 7.1$ on three degrees of freedom, against
$0.3$ on the original three-instance grid, whose
$1.195 \pm 0.011$ the full-grid summary supersedes). On all eighteen $n = 16$
instances the summary reads $1.224 \pm 0.018$ with
$\chi^2/\mathrm{dof} = 8.7$: the catch-up moves every full-grid
statistic away from the fixed base. The local decay rate grows
monotonically across the
grid: base $1.16$ to $1.28$ on the $18$-instance grid
($n \le 14$), to $1.31$ only through the four-instance $n = 16$
interval, the grid's most truncated ($1.32$ on all eighteen). The
steepening stands on the $18$-instance grid alone, without
extrapolation and without the $n = 16$ point: the log-step of the
mean-of-best per two qubits rises from $0.296 \pm 0.034$ to
$0.487 \pm 0.055$ between the first and third intervals, a
$3.0$-standard-error increase, and a fixed-base law fitted to
$n \le 14$ is rejected by the data it fits ($p = 0.026$ under the
covariance-estimated exact test, Appendix~\ref{app:methods}). The
rejection is not the budget's: rerun to optimizer convergence under
the registered protocol, the same grid rejects the same law at
$p = 0.041$, base $1.200 \pm 0.008$, local bases $1.16$, $1.20$,
$1.27$ (\texttt{analysis/reach\_exact\_test.log},
\texttt{analysis/converged\_reach.log}). The two budgets are
compared through one estimator, read at the same $B_0 = 16$ and
fixed by the registration before these data existed.

\begin{figure}[H]
\centering
\includegraphics[width=\textwidth]{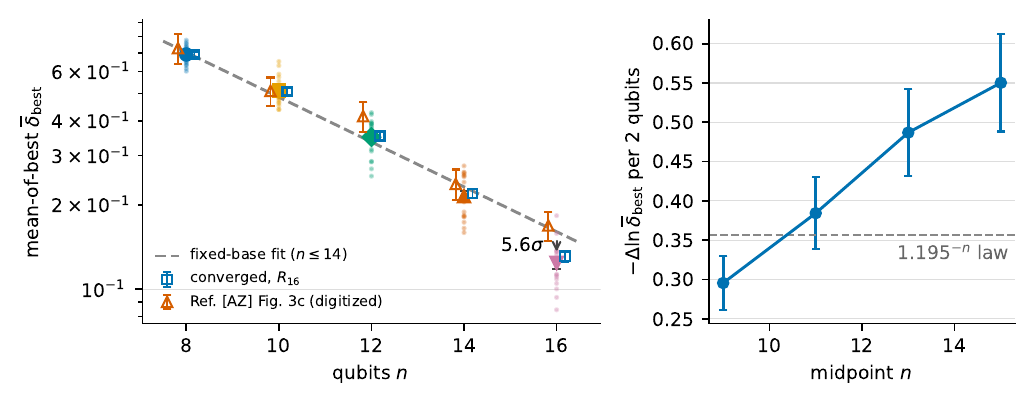}
\caption{\emph{Left:} mean-of-best peakedness versus $n$ ($18$
instances per size, the full grid including the registered
$n = 16$ catch-up; the published $n = 16$ values rest on the first
four instances, and the text reports both readings; small dots:
per-instance bests; error bars: standard error of the instance
mean). The dashed line is the weighted fixed-base fit to
$n \le 14$; the arrow marks the shortfall of the plotted
$n = 16$ point ($5.6\sigma$ on all eighteen instances), before the
step-cap correction of the text. Open squares: the registered
converged grid at the matched $B_0 = 16$ estimator, the level
with the step cap removed. Open triangles: the five
instance-averaged points of Ref.~\cite{aaronson2024peaked}
Fig.~3c, $\tau_p = \tau_r/2$, digitized from their published
raster (axes calibrated on the tick marks, each data point read at
the center of its marker's bounding box;
\texttt{analysis/az\_fig3c\_digitize.py});
their error bars are the digitization's reading uncertainty, the
half-height of the marker, not their statistical error, which the
raster does not resolve. Their series is instance-averaged with
the instance count stated for their panels a and b only, against
best-of-$200$ restarts on $18$ instances here, so the two series
measure the same regime under different search budgets.
\emph{Right:} log-steps of the
mean-of-best per two qubits; the horizontal line marks the fixed
$1.195^{-n}$ law fitted on the original three-instance grid.}
\label{fig:scaling}
\end{figure}
 
The measured $n = 16$ point falls below the fixed-base
extrapolation of the smaller sizes, where the step cap binds
hardest: correcting for the deficit of Appendix~\ref{app:controls}
moves it from $0.126$ to $0.139$ against an extrapolation of
$0.161$. On all eighteen instances the point reads $0.124$ raw
and $0.138$ corrected, against the same extrapolation, which rests
on $n \le 14$ and does not move. The shortfall below our own
extrapolation
cannot be read as a landscape property independently of the
truncation caveats that follow. On
$\ln\overline\delta$, the fitted quantity, the shortfall is $4.4$
standard errors of the instance mean raw, $2.5$ corrected;
propagating the extrapolation's own uncertainty lowers these to
$3.4$ and $2.0$. On all eighteen instances every figure grows:
$5.6$ raw and $3.4$ corrected on the instance mean, $4.1$ and
$2.4$ propagated. The deficit correcting the point rests on a
single instance, and the $1600$-step control that measures it is
itself truncated from $n = 10$ on (Appendix~\ref{app:controls}), so
the corrected shortfall is an upper bound on what convergence
would leave. All four published instances fall below the
extrapolation individually, the most favorable by $1.3$ sample
standard deviations on $\ln\delta$ over the four per-instance
bests; on all eighteen, seventeen fall below while the widened
instance spread carries the single most favorable one $0.7$ sample
standard deviations above it: the shortfall concerns the instance
mean, not every draw. At fixed protocol budget, no bare fixed base fits
the reachable peakedness over the measured grid. The wall shows
without a fit in the stalled fraction, the share of restarts
finishing below $(1-\varepsilon)$ of the batch best at
$\varepsilon = 0.1$: $0.005$, $0.13$, $0.44$, $0.62$ at
$n = 8$--$14$ and $0.77$ at $n = 16$ ($0.65$ on all eighteen),
with wide instance-to-instance spread at the larger sizes ($0.07$
to $0.99$ at $n = 14$).

The step cap does not produce this curvature; the statement is
measured, not argued: on the converged grid the registered
steepening statistic (the difference of the last and first
log-steps on $n \le 14$) reads $+0.171 \pm 0.064$ against
$+0.190$ at the frozen budget, so truncation accounts for ten
percent of the steepening and leaves a $2.7$-standard-error rise
(\texttt{analysis/converged\_reach.log}). An earlier defense
corrected the frozen statistic by the measured per-size deficits
of Appendix~\ref{app:controls} and rested on their log-linearity;
the registered test rejects that linearity, and the correction
argument is withdrawn as its registration requires.
Truncation does change the level, hence the $n = 16$ shortfall,
the deficit being largest exactly there. The converged grid now
carries its own $n = 16$ row, $18$ instances with the acceptance
criterion met (Appendix~\ref{app:controls}), and the shortfall
survives at a scale the residual truncation does not approach. The
converged point reads $0.1312 \pm 0.0061$ at the matched
$B_0 = 16$ estimator against $0.1635$ extrapolated from the
converged $n \le 14$ fit, $19.8\%$ below, $4.7$ standard errors of
the instance mean and $3.5$ with the extrapolation's own
uncertainty propagated, local base of the last interval $1.295$,
full-grid steepening $+0.215 \pm 0.071$
(\texttt{analysis/converged\_reach.log}). What these data exclude
is a bare fixed base, not every exponential reading: a fixed base
under a growing prefactor, $\delta \propto n^{2.8}\,1.58^{-n}$,
reproduces the measured rise of the local base within the quoted
interval errors and is not excluded by four intervals. The
$n = 50$ estimate of Ref.~\cite{aaronson2024peaked} fails under
either reading, since a base of
$1.58$ also destroys it.

\subsection{Restart budgets}
\label{sec:budget}

\begin{figure}[H]
\centering
\includegraphics[width=\textwidth]{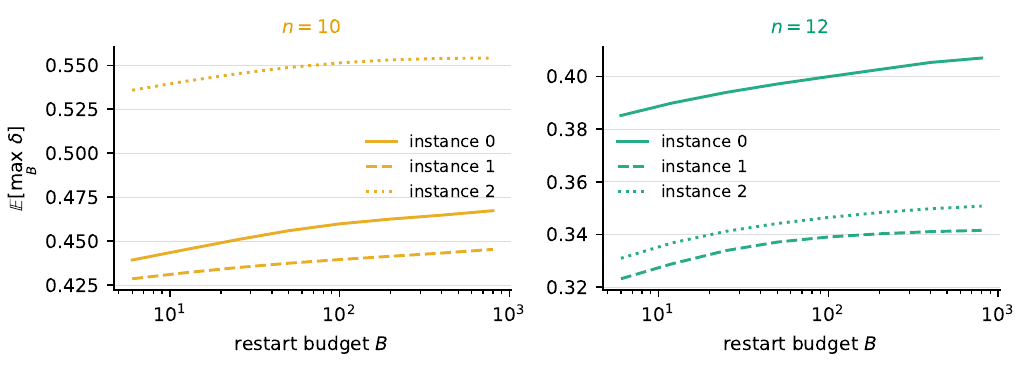}
\caption{Expected best-of-$B$ curves (the without-replacement
estimator of Eq.~\eqref{eq:bestofb}
over the $800$-restart ensembles) for three instances each at
$n = 10, 12$. The log-budget slope persists to $B = 800$ at
$0.61$--$0.63\times$ its own $[25, 200]$ reference, $\approx 9\times$
the i.i.d.-exponential null at $n = 12$ on these ensembles, and the
maximum stays unique (degeneracy $1$--$2$ per $800$ at $10^{-4}$).}
\label{fig:budget}
\end{figure}
 
Against the wall being a restart-budget artifact, the
$800$-restart ensembles (Fig.~\ref{fig:budget}) show the
log-budget slope of $R(n, B)$ persisting at the same order to
$B = 800$. The comparator is the i.i.d.-exponential null, in which
best-of-$B$ is the maximum of independent Porter--Thomas cell
values of mean $2^{-n}$, so the log-budget slope is exactly
$2^{-n}$. Best-of-$B$ is estimated exactly, without replacement,
from the ordered final values
$\delta_{(1)} \le \dots \le \delta_{(N)}$ of an ensemble,
\begin{equation}
\Ee\big[\max \text{ over a random } B\text{-subset}\big] =
\sum_{i} \delta_{(i)}\, \binom{i-1}{B-1} \Big/ \binom{N}{B}.
\label{eq:bestofb}
\end{equation}
On the consolidated grid (fit window $B \in [25, 200]$) the median
measured slope is $14\times$ this null at $n = 12$, $51\times$ at
$n = 14$ and $174\times$ at $n = 16$ ($111\times$ on all
eighteen). On the fresh $800$-restart ensembles the pooled
$[100, 800]$ slope is $0.63$ of the $[25, 200]$ slope of the
\emph{same} ensembles at $n = 12$ and $0.61$ at $n = 10$, still
$\approx 9\times$ the null at $n = 12$, with no re-emerging
degenerate maximum. (Against the $[25, 200]$ slope of the separate
$200$-restart grid: $0.69$ and $0.75$; the quoted ratio is within
one ensemble and one estimator.)
Brute restarts are inefficient in absolute terms: from $n = 10$
on, an e-fold of budget buys a few $10^{-3}$ of peakedness while
the level itself falls from $0.51$ at $n = 10$ to $0.12$ at $n = 16$, so
no restart budget holds a fixed fraction of the reach as $n$
grows.

\subsection{Optimizer robustness and the class boundary}
\label{sec:robustness}

On a fixed instance ($n = 10$), the entire
profile (best value within $2.8\%$ of the protocol's, all restarts
near-optimal, empty intermediate band, unique maximum) is
invariant under initialization scale $\times 10$ and learning rate
$\times 4$. Plain SGD, in the stable-local class as registered,
reaches $0.099$ against $0.463$ at the protocol's $400$ steps,
with $16/200$ restarts near-optimal and the control matrix's only
non-empty intermediate band (mass $0.275$):
registered
prediction 3 is falsified here; Appendix~\ref{app:registration}
records
the verdict. The mechanism is pacing, not the landscape: on twelve
paired restarts from identical initializations, the two optimizers
see the same median gradient norm along the trajectory ($3.9$ and
$4.0\times10^{-2}$) while their median per-step displacements
differ by a factor $47$ ($6.9\times10^{-2}$ against
$1.5\times10^{-3}$): raw-gradient-paced descent travels $0.7$ in
parameter space over the run where the protocol travels $38$. This
measurement is what the pacing criterion of the amended class
(Sec.~\ref{sec:protocol}) abstracts. An exploratory, unregistered
control gives SGD the matched $13{,}200$-evaluation budget at
$B = 16$ on the same instance: its reach more than doubles
($0.084 \to 0.200$ at the matched restart read) yet rides the cap
on all sixteen restarts and lands at $0.44$ of the matched Adam
cell ($0.4514$); budget explains part of the failure, pacing the
rest, and the failure stays downward at $33\times$ the budget.

Whether a distinct in-class optimizer reaches further than the
protocol's Adam is the upward question, the one the hardness
claims stand on, and it is answered by the preregistered
optimizer-class block (Appendix~\ref{app:methods}): optimizer
$\in$ \{Adam, L-BFGS-B\} crossed with initialization $\in$
\{normal, Haar\}, three instances per size at the matched
$13{,}200$-evaluation budget, the reach ratio
$\rho_{\rm arm}(n) = R_{\rm arm}(n)/R_{\rm conv}(n)$ read at
matched restart budget $B_0 = 16$, each outcome's cost fixed
before the data existed. Both Haar-initialized arms sit at or
below unity at every size, so the ceiling is not the normal
initialization's. Read as a two-by-two, the excess is an
optimizer-by-initialization interaction: at Haar initialization the
two optimizers coincide within errors at the largest size
(\texttt{analysis/optclass\_reach.log}). The (L-BFGS-B, $\sigma$)
arm is unity within
errors through $n = 14$ and exceeds it at the largest size:
$\rho(16) = 1.0389 \pm 0.0156$ on the paired construction,
per-instance ratios $1.069$, $1.031$, $1.017$
(Table~\ref{tab:optclass}; the denominator there is the converged
cell on instances $0$--$2$, $0.1388 \pm 0.0077$, not the
eighteen-instance $0.1312$ of Sec.~\ref{sec:steepening}). The
$n \le 14$ arms and the $n = 16$ cells ran on different
architectures; the dedicated control of
Appendix~\ref{app:registration} closes that confound. The
exceedance, stated exactly: one
departure above unity, at the largest size only, on three
instances, two degrees of freedom, where two standard errors give
$82\%$ two-sided confidence rather than $95\%$
($t_{0.975,2} = 4.30$). Because every arm reaches each instance's
atom at $n = 8$ to within $10^{-4}$
($\rho(8) = 1.00003 \pm 0.00002$ paired), the registered growth
statistic $\Gamma = \ln\rho(16) - \ln\rho(8)$ collapses to
$\ln\rho(16)$ up to that anchor: its significance is the $n = 16$ cell's, not
additional evidence, and the anchor-free weighted slope over
$n = 10$--$16$, $+0.0016 \pm 0.0016$, does not distinguish growth
from flat. The two registered constructions part company here for
the first time: the paired statistic triggers the growing-excess
branch, the literal one does not (Table~\ref{tab:optclass}), and
the tie-break fixed before these data existed selects the costlier
reading (Appendix~\ref{app:registration}). The registered outcome
is
therefore the third: the third falsifier of C-hardness is
triggered by our own measurement (Sec.~\ref{sec:hardness}). The
arm diagnostics are clean:
every restart of every arm peaks on $0^n$, and the L-BFGS-B arms
terminate by scipy convergence on all $48$ restarts per size at
$n \le 14$. The margin is a few percent with a shallow measured
trend: extrapolating the registered statistic's own rate, the
quasi-Newton excess reaches a factor $2$ near $n \approx 150$ and
a factor $\approx 1.2$ at the $n = 50$ of
Ref.~\cite{aaronson2024peaked}; the anchor-free slope
puts both further out. The excess reshapes the hardness
bookkeeping, not the reach law: L-BFGS-B's own reach still falls
from $0.72$ to $0.14$ over the grid at the same estimator, local
base $1.31$ on the last interval against Adam's $1.33$ on the same
three instances. At their $n = 50$ the margin
is worth less than one qubit of the reach law.

\begin{table}[H]
\caption{The optimizer-class block: reach ratio
$\rho_{\rm arm}(n) = R_{\rm arm}(n)/R_{\rm conv}(n)$ at matched
instances ($0$--$2$), matched evaluation budget ($13{,}200$), and
matched restart budget $B_0 = 16$, through the exact
without-replacement estimator. The $\sigma$-initialized arm carries
the protocol's normal initialization,
$\theta_0 \sim \mathcal{N}(0, \sigma^2 I_P)$ with $\sigma = 0.1$
(Sec.~\ref{sec:protocol}). Both readings of the registered statistic
are shown: \emph{paired} forms the ratio per instance and averages it,
with the standard error over the three instances; \emph{literal}
is the ratio of the two arm means, with relative errors combined
in quadrature. $\Gamma = \ln\rho(16) - \ln\rho(8)$ is the registered
growth statistic; $\rho(8) = 1$ to within $3 \times 10^{-5}$ for
every arm, so $\Gamma$ is $\ln\rho(16)$ up to that shift
(Sec.~\ref{sec:robustness}). The
denominator, the (Adam, normal)
converged cell, is $R_{\rm conv} = 0.7205 \pm 0.0176$,
$0.4806 \pm 0.0352$, $0.3679 \pm 0.0193$, $0.2473 \pm 0.0125$,
$0.1388 \pm 0.0077$ at $n = 8$--$16$. The paired construction of
the (L-BFGS-B, $\sigma$) arm exceeds unity at two standard errors
at $n = 16$ and triggers the registered branch; the literal
construction does not, and the tie-break fixed in the analysis
code before these data existed selects the costlier reading
(Appendix~\ref{app:registration}).}
\label{tab:optclass}
\centering
\begin{tabular}{@{}llccc@{}}
\toprule
 & & L-BFGS-B, $\sigma$ init & L-BFGS-B, Haar init & Adam, Haar init \\
\midrule
$\rho(8)$  & paired  & $1.0000 \pm 0.0000$ & $1.0000 \pm 0.0000$ & $1.0000 \pm 0.0000$ \\
           & literal & $1.0000 \pm 0.0345$ & $1.0000 \pm 0.0345$ & $1.0000 \pm 0.0345$ \\
\addlinespace
$\rho(10)$ & paired  & $1.0056 \pm 0.0036$ & $0.9929 \pm 0.0091$ & $0.9805 \pm 0.0125$ \\
           & literal & $1.0056 \pm 0.1043$ & $0.9938 \pm 0.1079$ & $0.9811 \pm 0.1057$ \\
\addlinespace
$\rho(12)$ & paired  & $0.9988 \pm 0.0058$ & $0.9928 \pm 0.0065$ & $0.9643 \pm 0.0088$ \\
           & literal & $0.9985 \pm 0.0724$ & $0.9935 \pm 0.0781$ & $0.9652 \pm 0.0781$ \\
\addlinespace
$\rho(14)$ & paired  & $1.0078 \pm 0.0067$ & $0.9866 \pm 0.0166$ & $0.9822 \pm 0.0154$ \\
           & literal & $1.0082 \pm 0.0756$ & $0.9882 \pm 0.0820$ & $0.9836 \pm 0.0804$ \\
\addlinespace
$\rho(16)$ & paired  & $1.0389 \pm 0.0156$ & $0.9257 \pm 0.0507$ & $0.9265 \pm 0.0563$ \\
           & literal & $1.0395 \pm 0.0871$ & $0.9310 \pm 0.1110$ & $0.9319 \pm 0.1142$ \\
\midrule
$\Gamma$   & paired  & $+0.0381 \pm 0.0151$ & $-0.0772 \pm 0.0548$ & $-0.0764 \pm 0.0608$ \\
           & literal & $+0.0387 \pm 0.0906$ & $-0.0715 \pm 0.1241$ & $-0.0706 \pm 0.1273$ \\
\bottomrule
\end{tabular}
\end{table}
 \FloatBarrier
 \section{The barren plateau does not explain the reach}
\label{sec:bp}

The reach of Sec.~\ref{sec:dynamics} shrinks with $n$ for every
optimizer measured. The explanation on record is the barren
plateau \cite{mcclean2018barren}.

A static route to a head start would be a gradient
signal at initialization stronger than the value signal. Evaluated
at the near-identity initialization, the exact gradient covariance
of Sec.~\ref{sec:kernel} closes that route: the gradient-variance
enhancement over the Haar limit at the operating depth is
$1.12$--$1.16$, flat through $n = 12$ ($1.14$; four sizes),
numerically close to the square root of the value-variance
enhancement at $n = 8$ ($\sqrt{1.34} = 1.16$, from the exact
$m_2 = 1.164$). Mean, variance, and gradient statistics at
initialization are Haar-like to $O(1)$: any advantage the flow
exploits is acquired during optimization, not present at the
start. The initialization signal thus sits in the barren-plateau
regime, its variance at the $\Theta(2^{-2n})$ Haar scale predicted
by Eq.~(7) of Ref.~\cite{mcclean2018barren}, up to the
$1.12$--$1.16$ factor.

What the plateau implies for
the maximum depends on a model of the field.
Gaussianity is a modeling hypothesis;
Eq.~\eqref{eq:hypotheses} supplies the second-order amplitude data,
not the Gaussian law. The comparison of Lemma~\ref{lem:gaussian}
yields the structural
statement of the statics: every separable Gaussian model
consistent with those data caps expected peakedness at $\mathrm{poly}(n)\,2^{-n}$,
while the reachable peakedness of Sec.~\ref{sec:dynamics} decays,
over the measured grid, far more slowly than $2^{-n}$.

A Gaussian model in the amplitude carrying the data
of Eq.~\eqref{eq:hypotheses} determines the entire law of the field:
$\delta$ is exponential with mean $1/D$ at every point, so $m_k = 1$
for every $k$, and the covariance of the objective is forced to the
Gaussian form $\Ee[\delta\delta'] = (1 + F)/D^2$, a function of $F$
alone. Both consequences fail against exact computation at finite
depth. The kernel of Eq.~\eqref{eq:kernel} is not a function of
$(F, n)$ alone (Sec.~\ref{sec:kernel}); and $m_2 \neq 1$ at
$(n, \tau_r) = (8, 8)$ at every probe point computed, by $16.4\%$
at the initialization and $3.6\%$ at the reachable solutions
(Sec.~\ref{sec:probedep}). One such point suffices, so the
conclusion does not depend on which: the model fails at every
depth computed, and the failure changes order, $O(1/D)$ in the deep limit,
where the Haar floor gives $m_2 = D/(D{+}1)$, and $O(1)$ at the
operating depth. The Gaussian model consistent with the exact
amplitude data is thus falsified at finite $n$, with no asymptotic
step and no undetermined constant. Lemma~\ref{lem:gaussian} adds
that this falsification is \emph{necessary} for the reach law's
scaling form. The statement is one of order; the lemma's constant is
not explicit.

A depth decomposition localizes the excess of reachable
peakedness. Anchoring each
instance at $\tau_r = 4n$, where the exact $m_2$ sits at its Haar
floor ($0.985$ against $D/(D{+}1) = 0.9846$ at $n = 6$), gives a
deep floor $\delta_{\rm deep} = 0.275(8), 0.118(2)$ at
$n = 10, 12$ (means over the three deep-anchored instances),
instance-stable to $\sim 2\%$, compatible with
$\mathrm{poly}(n)\,2^{-n}$, and atom-free (top-$3$ spreads of
$10^{-3}$). At $n \ge 10$, under-reach thus extends into the
Porter--Thomas regime itself. The finite-depth excess, the reach
at $\tau_r = n$ over the reach at $\tau_r = 4n$, grows from
$\approx 1.8$ to $\approx 3.1$ (means over three instances)
between $n = 10$ and $12$: on the anchored sizes, the
exponential excess of reachable peakedness over the deep
floor is carried by the finite-depth factor, the quantity a theory
of this landscape must explain.

Two depths confirm independently that the amplitude data do not
suffice. Those data are exactly depth-independent
(Eq.~\eqref{eq:hypotheses} holds at every $\tau_r \ge 1$), so every
separable Gaussian model consistent with them carries the same
ceiling at all depths. The measured reach does not: anchoring the
same instances at
$\tau_r = 4n$ lowers it, by the factors measured above. No
depth-independent second-order amplitude description generates a
depth-dependent reach. What the optimizer exploits is therefore
invisible to Eq.~\eqref{eq:hypotheses}; it shows in the moments of
$\delta$, which Fig.~\ref{fig:moments} measures at the
initialization probe, Sec.~\ref{sec:probedep} bounding how far that
measurement travels.
 \section{Geometry of the solution set}
\label{sec:solgeometry}

The solutions the protocol reaches carry an exact local structure,
the gauge of the parametrization and of the objective.
String-method scans measure the global structure of that set.

\subsection{The gauge, counted}
\label{sec:gaugecount}

Counting the internal wire segments of $V$ gives the dimension of
the continuous gauge of Sec.~\ref{sec:gauge} exactly.

\begin{proposition}[Gauge counting]
\label{prop:gauge}
Let $S$ be the number of internal wire segments of $V$; for the
convention of Sec.~\ref{sec:field}, $n$ even and the first peaking
layer covering all $n$ qubits,
$S = (n{-}2)(\tau_p{-}1) + 2\lfloor(\tau_p{-}1)/2\rfloor$. At
every $\theta$ at which the word map is regular at each gate and
the generic-position condition of Appendix~\ref{app:proofs} holds,
the fiber $\{\theta' : V(\theta') = V(\theta)\}$ contains a smooth
$3S$-dimensional family through $\theta$, and the corank of the
projective Jacobian of the output state is at least
$\max\bigl(3S,\, P - (2^{n+1}{-}2)\bigr)$. At the operating
configuration $(n, \tau_p) = (8, 4)$: $3S = 60$.
\end{proposition}

The bound is saturated: the measured corank equals it exactly at
the archived solutions of three sizes, $60$, $108$ and $162$ at
$(n, \tau_p) = (8, 4), (10, 5), (12, 6)$, and at Haar-random points
of twelve configurations spanning $(4, 2)$ to $(16, 8)$, depth scans
included, the excess over $3S$ appearing only as the trivial
state-space cap at $(4, 2)$; the same computation exhibits the
family constructively, re-solving the dressed word at fixed $V$ to
machine precision (\texttt{analysis/gauge\_dimension.log}). The
parametrization therefore carries exactly $3S$ continuous directions
fixing the output state $V U_r \ket{0^n}$.

The objective sees less of $\theta$ than the output state does. By
Eq.~\eqref{eq:field}, $\delta$ depends on $\theta$ only through the
ray of the probe $\ket{w(\theta)}$, so every direction fixing that
ray is exactly flat whatever it does to the output state: the gauge
of the objective is strictly larger than that of $V$.

\begin{proposition}[Boundary gauge of the objective]
\label{prop:probegauge}
Call a qubit \emph{closed} at a gate of $V$ if no later gate of
$V$ acts on it, and let $B_2$ and $B_1$ count the gates of $V$ with
two and with one closed qubit. For the convention of
Sec.~\ref{sec:field} with the first peaking layer covering all $n$
qubits: $B_2 = n/2$, $B_1 = 0$ for $\tau_p$ odd and
$B_2 = n/2 - 1$, $B_1 = 2$ for $\tau_p$ even. At every $\theta$ at
which the word map is regular at each gate, the corank of the
projective Jacobian of the probe ray
$\theta \mapsto [\,\ket{w(\theta)}\,]$ is at least
$3S + 9B_2 + 4B_1$. At the operating configuration
$(n, \tau_p) = (8, 4)$: $60 + 27 + 8 = 95$.
\end{proposition}

This bound is saturated as well. The measured probe-ray corank
equals $3S + 9B_2 + 4B_1$ at Haar-random points of the same twelve
configurations (the probe involves no instance, so the count holds
instance-free) and at the archived solutions of three sizes:
$95$, $153$ and $215$ at $(8, 4), (10, 5), (12, 6)$
(\texttt{analysis/probe\_gauge.log}). The
proposition bounds the rank everywhere, so equality
at one point forces constant rank on a neighborhood, and the
constant-rank theorem makes the fibers of
$\theta \mapsto [\,\ket{w(\theta)}\,]$ there smooth
$(3S + 9B_2 + 4B_1)$-dimensional submanifolds along which $\delta$ is
\emph{exactly} constant, to all orders
(Appendix~\ref{app:proofs}); Sec.~\ref{sec:geometry} reads the
Hessian at the solutions against this count. At $\tau_p = 1$ the
proposition recovers the brick-product geometry of
Theorem~\ref{thm:firstrung}: the probe ray ranges over a manifold
of dimension $3n$, measured as rank $24$ at $n = 8$.

The plainest piece of this boundary gauge acts on the objective but
not on the state: dressing the closing gates of $V$ by
$\bigotimes_j R_{z,j}(\alpha_j)$ sends the probe to
$e^{i \sum_j \alpha_j / 2} \ket{w}$, a phase, so $\delta$ is exactly
invariant while $V U_r \ket{0^n}$ changes. Every solution thus
carries a continuous $n$-parameter family of equal-$\delta$
solutions whose floor-normalized mutual overlaps $\hat q$ sit near
zero. All $n$ directions move
the state: the generators $Z_j$ are traceless, so no combination of
the dressing acts as an overall phase on the ray, and the rank of
their projected state motions is measured as $n$ at the solutions of
three sizes, smallest singular value $0.43$ at $n = 8$
(\texttt{analysis/probe\_gauge.log}). The solution set at fixed
$\delta$ therefore contains $n$ state-moving exactly flat directions
by symmetry alone (the floor of the $9B_2 + 4B_1$ boundary count, not
its total), so a pair with $\hat q \approx 0$ may be genuinely unrelated
or related by the boundary dressing alone.

\subsection{Geometry at the atom: the objective's gauge, in
closed form}
\label{sec:geometry}

Hessians at the atom solutions of the operating configuration
$(n, \tau_p) = (8, 4)$, $P = 210$, show a flat band of exactly $95$
eigenvalues at the numerical floor. The solutions are polished from
the protocol's stop to stationarity
($\|\nabla\delta\| \approx 3 \times 10^{-8}$), a diagnostic step
outside the frozen protocol, on which no reported reach level rests.
At every atom of the decorrelated set the count is invariant from
threshold $10^{-7}$ to $10^{-3}$: the largest band eigenvalue stays
below $1.4 \times 10^{-8}$, and the smallest of the remaining $115$
exceeds $1.3 \times 10^{-3}$, five orders of magnitude across the gap.
Ninety-five is the closed-form count $3S + 9B_2 + 4B_1$ of
Proposition~\ref{prop:probegauge}, the measured probe-ray corank at
every solution, and the identification holds direction by direction:
along each of the $95$ directions the probe ray is stationary to
first order,
speed $\le 3 \times 10^{-7}$ across the atom set, while along each of
the other $115$ it moves five orders of magnitude faster, $\ge 4 \times
10^{-2}$. The identity repeats at the polished best restarts of $(10, 5)$ and
$(12, 6)$: bands of exactly $153$ and $215$ against the same count
(\texttt{analysis/probe\_gauge.log}).

At every solution computed, the band is the objective's gauge and
nothing else. Its
decomposition is structural, not fitted: $3S = 60$ directions fix $V$
and with it the output state (Proposition~\ref{prop:gauge}); the
remaining $9B_2 + 4B_1 = 35$ move the output state while fixing the probe
ray, tracing an exact equal-$\delta$ orbit through every solution
($35$, $45$ and $53$ state-moving flat directions at the three sizes,
measured as ranks of the projected state motion rather than read off
as differences, of which the boundary $R_z$ dressing of
Sec.~\ref{sec:gaugecount} supplies $n$). The flatness of the counted
directions is exact to all orders (Appendix~\ref{app:proofs}); that
the band holds nothing else is measured, not derived. Nothing in
the flat band constrains $\delta^*$: what fixes the ceiling is the
transverse spectrum of the $115$ non-flat directions, which this paper
measures at the solutions but does not model.

The count also explains the band at the protocol's own stop, governed
by learning-rate decay at $\|\nabla\delta\| \approx 10^{-3}$. There the
gauge directions are not exactly flat: in adapted coordinates their
Hessian entries are first order in the gradient. The band's
residual scale measures a constant fraction $\approx 4 \times 10^{-2}$
of $\|\nabla\delta\|$ across the four atoms
(\texttt{analysis/hessian\_solutions.log}, whose corank column is the
output-state count $3S = 60$ of Proposition~\ref{prop:gauge}). The
count's threshold dependence at that stop is this stall artifact, which
polishing removes without moving $\delta$ at the atom ($0.752655$ at
instance $0$, unchanged at six decimals across the four maximizers).

This sharpens the local reading of the solution \emph{shelf} that
Sec.~\ref{sec:shelf} probes globally: through
every solution runs an exact equal-$\delta$ orbit of dimension
$9B_2 + 4B_1 = \Theta(n)$ that moves the state (room enough for
decorrelated maximizers and for drift without loss). Whether the
reachable set exceeds this orbit is a question about paths, not
Hessians, answered at the trap scale in Sec.~\ref{sec:shelf}.

\subsection{The reachable ensemble: a single corrugating shelf}
\label{sec:shelf}

Reachable solutions are pairwise decorrelated at every size probed
(Fig.~\ref{fig:pq}) and insensitive to the initialization scale
($\sigma = 0.1 \to 1.0$ control). Recorded
trajectories at $n = 12$ show the objective converging long before
the state commits: the state's commitment time (first crossing of
overlap $0.5$ with its own endpoint) comes a factor $1.9$ after
the value reaches $90\%$ of its final level and a factor $8.4$
after the $50\%$ crossing (medians over $8$ recorded restarts). At
$n = 8$ the same instrument, on $12$ restarts, gives the opposite
ordering: commitment arrives at a factor $0.86$ of the $90\%$
crossing, before it. The separation is thus a property of the
larger size, where the shelf is wide and the optimizer drifts
along it, not a mechanism visible at $n = 8$, where the atom is
reached directly. Two sizes and small samples make this suggestive
only: drift on a nearly flat set rather than convergence into the
distinct traps of Ref.~\cite{anschuetz2022traps}; the connectivity
measurement below, not this one, carries the conclusion.

\begin{figure}[H]
\centering
\includegraphics[width=\textwidth]{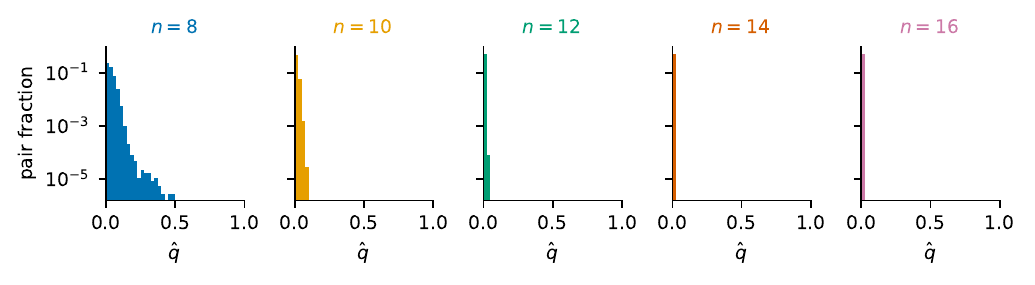}
\caption{Floor-normalized pair-overlap distributions $p(\hat q)$,
pooled over all instances (log scale; $\varepsilon = 0.2$ filter).
At every size probed the distribution is a spike at $\hat q = 0$
with a rapidly decaying tail: the intermediate band
$\hat q \in [0.25, 0.75]$ carries $32$ of $3.6\times10^{5}$ pairs
at $n = 8$ ($9\times10^{-5}$) and zero pairs at every
$n \ge 10$, including zero of $2.5\times10^{5}$ pairs at $n = 16$
over all eighteen catch-up instances.}
\label{fig:pq}
\end{figure}
 
Whether that set is shattered or connected is settled by a string
method \cite{e2002string,e2007simplified}, in the simplified
equal-arc-length form of Ref.~\cite{e2007simplified}: a
fixed-endpoint polyline of waypoints between two decorrelated
near-optimal solutions, alternately ascended on $\delta$ and
reparametrized to equal arc length, its minimum read on a dense
grid and accepted only under a path-validity criterion
(Appendix~\ref{app:controls}). Two
baselines calibrate it: raw segments (unoptimized straight
interpolations) and scrambled controls (the same protocol run
toward a near-optimal solution of a different instance, which must
fail for the instrument to count). At every size probed the
re-optimized paths stay $10^2$--$10^3$ times above the trap scale
$2^{-n}$ (median $1.3\times10^2$ at $n = 8$ rising to
$1.6\times10^3$ at $n = 16$) while raw segments and controls
collapse to it or below: no barrier between reachable solutions
descends to the trap scale, so the reachable ensemble does not
fragment into disconnected components (an ``archipelago'') at that
scale. The stricter criterion fixed in the analysis code, a path
minimum at or above $0.8$ of the endpoint level matching the
$\varepsilon = 0.2$ near-optimality filter, is met by the per-size
statistic at no size and by single paths only at $n = 8$.
The measurement establishes, and we use below, connectivity at the
trap scale, not at the near-optimal level. What changes with $n$
is the depth of the shelf's corrugation, quantified by the
corrugation floor $\varrho(n)$, the median string-path minimum
relative to the endpoint level; smaller $\varrho$ means deeper
corrugation (Fig.~\ref{fig:corrugation}).

A first pass measured $\varrho$ at four sizes, one instance each,
at two string resolutions; re-running the one size where
resolution was tested lowered the median $\varrho$ by $15.6\%$, a
shift of the same order as the effect, so that design could not
carry a statement about functional form. We therefore repeat the
measurement on $180$ string-method paths at one matched resolution
of $64$ segments across $n = 8$--$16$, three instances per size
and twelve disjoint pairs each, under an analysis registered
before those runs existed. The corrugation floor, a mean over the
three instance medians with its standard error over instances, is
\begin{equation*}
\varrho(n) = 0.726 \pm 0.025,\; 0.560 \pm 0.005,\; 0.444 \pm 0.014,\;
0.310 \pm 0.027,\; 0.226 \pm 0.017
\end{equation*}
at $n = 8, 10, 12, 14, 16$. Weighted by those errors, an
exponential $\varrho \propto 0.871^n$ fits with $\chi^2 = 3.4$ on
three degrees of freedom ($p = 0.34$), while the best fixed power
law, $\varrho \propto n^{-1.50}$, reads $\chi^2 = 14.4$:
asymptotically $p = 0.002$, but the weights are estimated on three
instances, and the same statistic calibrated by simulation at that
sample size gives $p \approx 0.1$, $0.06$ under pooled weights
(\texttt{analysis/corrugation\_calibration.log}); the interval
exponents $\ln(\varrho_i/\varrho_{i+1})/\ln(n_{i+1}/n_i)$ rise
$1.16, 1.28, 2.32, 2.38$. The first condition of the registered
verdict\footnote{Registered statistic: the median of $\varrho$ over
the twelve pairs per (size, instance), then the mean of the three
instance medians, with the standard error over instances as the error
bar of record. Per-instance medians: $0.713$, $0.690$, $0.774$ at
$n = 8$; $0.571$, $0.553$, $0.557$ at $10$; $0.472$, $0.433$,
$0.426$ at $12$; $0.309$, $0.357$, $0.264$ at $14$; $0.237$,
$0.193$, $0.248$ at $16$.
Registered test: weighted least squares of $\ln\varrho$
against $n$ and against $\ln n$, compared by weighted $\chi^2$ on
three degrees of freedom, with the interval-exponent test for a single
fixed exponent. Registered verdict: if the exponential attains the
lower $\chi^2$ \emph{and} a fixed exponent is excluded at $p < 0.05$,
the claim that the corrugation deepens faster than any fixed power
stands; otherwise it is withdrawn from the abstract and replaced by
``the corrugation deepens with $n$; these data do not determine the
functional form,'' and Conjecture~\ref{conj:shelf} drops its rate
clause.} is met and the second does not survive its calibration, a
fixed exponent not being excluded at $p < 0.05$; the registered
fallback therefore applies, the costlier reading: the corrugation
deepens with $n$, and these data do not determine the functional
form. The errors carrying this reading are instance-level, not
within-instance bootstrap. The abstract, condensed after this
verdict, carries neither the original claim nor the replacement
sentence; the fallback reading stands here and in
Conjecture~\ref{conj:shelf}.

\begin{figure}[H]
\centering
\includegraphics[width=0.47\textwidth]{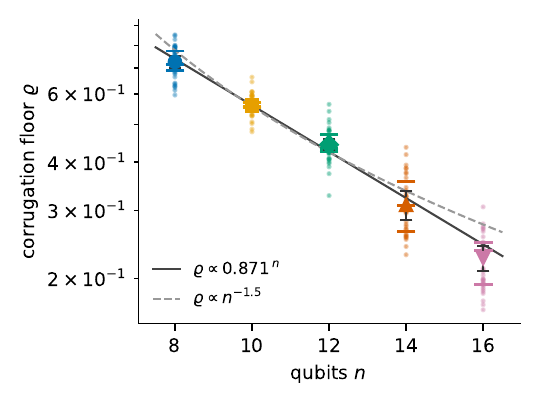}
\caption{Corrugation floor of the connected shelf, from the registered
sweep: $180$ string-method paths,
three instances per size, twelve disjoint pairs each, all at one
resolution of $64$ segments. Small dots are per-pair minima relative
to the endpoint level, dashes are the three instance medians, large
markers their mean, error bars the standard error over instances. The
exponential $\varrho \propto 0.871^n$ (solid) fits those errors with
$\chi^2 = 3.4$ on three degrees of freedom; the best fixed power law
(dashed) reads $\chi^2 = 14.4$, a comparison whose calibration
Sec.~\ref{sec:shelf} discusses. Controls
and raw segments (not shown) collapse to the trap scale
$2^{-n}$ at every size probed.}
\label{fig:corrugation}
\end{figure}
 
Neither of the following belongs to the registered verdict. The
exponential is adequate, not merely the better of two models
($\chi^2 = 3.4$, $p = 0.34$), with one measured threat to that
adequacy: the consolidated sweep is least resolved at the top of
the grid, adjacent-waypoint overlap falling from
$0.894$ to $0.667$ over $n = 8$--$16$
(Appendix~\ref{app:controls}), so the top-of-grid points carry the
largest resolution uncertainty. The base is stable against the change of design: the sweep
moved individual points by $2$ to $7\%$ while leaving the fitted
base at $0.871$, the newly measured $n = 14$ falling on the curve.
 \clearpage
\section{Shrinking reach without a trap-scale overlap gap}
\label{sec:hardness}

At every level the fixed protocol reaches, direct measurement refutes
the trap-scale overlap-gap picture of Ref.~\cite{gamarnik2021overlap}
for this landscape: feasible paths connect decorrelated
near-optimal solutions at all five sizes (Sec.~\ref{sec:shelf}),
and the corrugation along them, while deepening, never approaches the
trap scale $2^{-n}$. The pairwise decorrelation of Fig.~\ref{fig:pq}
does not carry this conclusion alone: it bounds the number of basins
from below, but a sparse sample of many mutually decorrelated
clusters would empty the intermediate band just as effectively.
Connectivity is what excludes fragmentation, and a path search is
one-sided, so what it excludes is fragmentation at the trap scale,
not clustering at the near-optimal level. Hardness, if structural
here, cannot be clustering hardness at the trap scale.

We conjectured the obstruction to be a triple, stated for the
brickwall ensemble at the operating regime
$(\tau_r, \tau_p) = (n, n/2)$. Two members stand; the third was
withdrawn by its own registered falsifier and is kept below as
stated, so that what was withdrawn is inspectable.

\begin{conjecture}[C-shelf]
\label{conj:shelf}
The near-optimal set reachable by stable-local optimization is, for
every instance and every $n$, path-connected at the trap scale, so
that no overlap gap forms at that scale. Over the sizes probed the
corrugation floor $\varrho(n)$ along those paths deepens with $n$,
these data not determining the functional form; persistence of the
deepening beyond $n = 16$ is the conjectural part.
Clustering at the near-optimal level is not excluded.
\end{conjecture}

\begin{conjecture}[C-reach]
Restarts are exponentially inefficient substitutes for signal: the
budget required to maintain a fixed fraction of the ceiling
$\delta^*$ grows as $\ln B = \Omega(n)$. The fraction is kept
qualitative because $\delta^*$ is measured only at $n = 8$, and
there from below; the
budget experiments measure directly the growth of reach with
$\ln B$.
\end{conjecture}

\begin{conjecture}[C-hardness, withdrawn]
The landscape defeats stable-local algorithms: no statistic
available to the flow (value, gradient, or their finite histories)
carries enough signal to beat the shrinking reach, even though the
solution set remains connected.
\end{conjecture}

The picture carried three registered falsifiers, one per conjecture.
The third, an optimizer class whose reach ratio to the protocol's
exceeds unity and grows, was triggered by our own measurement: the
(L-BFGS-B, $\sigma$) arm of the preregistered optimizer-class block
exceeds converged Adam at the largest size
(Sec.~\ref{sec:robustness}, Table~\ref{tab:optclass}). Per the
registration's own outcome table
(Appendix~\ref{app:registration}), C-hardness is
\textbf{withdrawn, not weakened}: the falsifier moves from one this
paper proposes to one it reports. What the withdrawal leaves
untouched: the reach of every arm of the optimizer-class block still
shrinks with $n$ at an essentially unchanged rate
(Sec.~\ref{sec:robustness}), so the measured obstruction to
generation stands as a property of the class-so-far, while the
claim that \emph{no} in-class statistic beats the shrinking reach
does not survive a member beating the best previously measured one by
a growing margin.
 \section{Relation to prior work}
\label{sec:related}

On the fit's own error the fixed-base summary of
Sec.~\ref{sec:steepening} sits $4.6$
standard errors above the $1.189$ per qubit of
Ref.~\cite{aaronson2024peaked}, $1.7$ once rescaled by the misfit
($5.9$ and $2.0$ on all eighteen): on this grid and regime the
frozen protocol decays faster, hence reaches less, than their fit,
which averages over a curved trend rather than fixing a base the
data sustain.

Ref.~\cite{aaronson2024peaked} reports an
instance-averaged value at $n = 16$ in the same regime (their
Fig.~3c), a figure point without step budget or tabulated value,
so the comparison is carried two ways. Their fitted law evaluates
to $\approx 0.18$ at $n = 16$
($5\times10^{-4} \times 1.189^{34}$, from their Sec.~3 base and
$n = 50$ estimate), above our raw $0.126$ and corrected $0.139$ on the four
published instances.
Their five published points are digitized onto
Fig.~\ref{fig:scaling} directly (reading method and uncertainty in
the caption): the digitized $n = 16$ point reads
$0.169 \pm 0.021$, above our frozen $0.124$, corrected $0.138$,
and converged $0.131$ alike, all on the eighteen instances. At
$n = 8$, $10$, and $14$ the reading error covers our levels at both
budgets (the $n = 14$ point: $0.238 \pm 0.029$ against our $0.215$
frozen, $0.220$ converged), while at $n = 12$ their digitized point
sits above them. The level statement is therefore made at $n = 12$
and $n = 16$: there their protocol reports reaching more than ours
does at any budget we measured.

Their Sec.~3 describes Adam, their released code
calls L-BFGS-B, and both are members of the class as
registered, compared head to head in the optimizer-class block
of Sec.~\ref{sec:robustness}, where the quasi-Newton member ends a
few percent above
converged Adam at the largest size. The block compares those two
members, not their full pipeline: further schedule tuning or a
sequential warm start is untested here, and a warm start, its
restarts dependent, sits outside the class as registered.
Sec.~\ref{sec:robustness} sizes the quasi-Newton margin against
their $n = 50$ target. In the neighboring regime
$\tau_r = n/2$
they report sizes to $n = 24$ fitted by a fixed exponential (their
Fig.~3d); whether the steepening measured here at $\tau_r = n$
appears there is untested.

The shallow-depth computation of Sec.~\ref{sec:atom} bears on the
cube-root
conjecture of Ref.~\cite{aaronson2024peaked} (their Conj.~3.2): on our
grid the fitted exponent is not depth-independent, falling from $0.28$
to $0.055$ over $\tau_p = 1$--$4$ at $n = 8$, while the value $1/3$ at
$\tau_p = 1$ is not excluded; Appendix~\ref{app:cuberoot} details the
reading tested and its caveats.

The overlap-gap program derives algorithmic hardness from a
topological property (near-optimal solutions cluster, with a
forbidden band of intermediate overlaps) for algorithms \emph{stable}
in a specific sense: the output solution moves by at most $\kappa$
when a small part of the instance is resampled along an interpolation
path (Ref.~\cite{gamarnik2021overlap}, Sec.\ ``OGP is an obstruction
to stability''; the classes certified stable there include low-degree
polynomials, approximate message passing, local algorithms, and
QAOA). The stable-local class of Sec.~\ref{sec:protocol} is a
constraint on optimizer dynamics rather than on instance
sensitivity; no optimizer run here (Adam, plain SGD,
L-BFGS-B, the Haar-initialized arms of Sec.~\ref{sec:robustness}) has
been tested for $\kappa$-stability, and the formal link between the
two notions is open. Sec.~\ref{sec:hardness} reads the connectivity
measurement against this program.

Sec.~\ref{sec:shelf} confronts the trap-based obstructions of
variational landscapes \cite{anschuetz2022traps} with recorded
trajectories and connectivity scans. Sec.~\ref{sec:bp} reads the
barren-plateau prediction of Ref.~\cite{mcclean2018barren} against
the initialization statistics.
 \section{Limitations}
\label{sec:limits}

The equality branch of the first-rung theorem is a density
hypothesis tested numerically, not proved, bounding the exact
maximum at one peaking layer (Sec.~\ref{sec:atom}). The formal
link between the stable-local class and the overlap-gap program's
stability notion is open, which bounds the synthesis
(Sec.~\ref{sec:hardness}). The pacing criterion is silent on
quasi-Newton line search, so the amended class's boundary stays
untested behind the optimizer result (Sec.~\ref{sec:protocol}).

Every result holds for noiseless circuits on one brickwall
ensemble under the fixed parametrization, at the operating regime
unless stated, on the variational manifold only; every reach claim
is quantified over the stable-local class (Sec.~\ref{sec:setup}).
The deep-limit ceiling binds searches of the manifold, not the
generation problem; its separation from the measured reach is one
of scaling form, not magnitude (Sec.~\ref{sec:ceiling}). The
entanglement budget is vacuous at the operating peaking depth and
bites below it, bounding the finite-depth ceiling
(Sec.~\ref{sec:deltastar}). The finite-depth ceiling is measured
at $n = 8$ only, from below, anchoring that result at one size
(Sec.~\ref{sec:atom}). Measured over sizes spanning a factor of
two, the data exclude a bare fixed base, not every exponential
reading: a growing prefactor over a fixed base survives four
intervals (Sec.~\ref{sec:steepening}). The plateau's presence is
measured through $n = 12$, on four sizes (Sec.~\ref{sec:bp}).

At $n \ge 10$ the frozen-budget reach is a truncated optimizer's
(Sec.~\ref{sec:dynamics}). Truncation accounts for ten percent of the
steepening and for part of the $n = 16$ level; the fixed-base
rejection stands at convergence (Sec.~\ref{sec:steepening}). A
vanishing floor-normalized overlap can come from the boundary
dressing alone, and that ambiguity bounds the pairwise
decorrelation of the connectivity result (Sec.~\ref{sec:shelf}). The magnitude of the
moment excess does not travel beyond the initialization probe and
bounds, rather than characterizes, the statics result elsewhere
(Sec.~\ref{sec:probedep}). The exceedance behind the optimizer
result rests on three instances at the largest size, two standard
errors giving $82\%$ rather than $95\%$ two-sided confidence
(Sec.~\ref{sec:robustness}).

A path search is one-sided: it excludes fragmentation at the trap
scale, not clustering at the near-optimal level, and this
one-sidedness bounds the connectivity result
(Sec.~\ref{sec:shelf}). The string resolution uncertainty is
largest at the top of the grid, so the connectivity result's
path levels there are the least resolved (Sec.~\ref{sec:shelf}).
Truncated moment evaluations carry a computed error certificate,
below $1\%$ on the plotted moments and $2.3\%$ of the fourth
moment on the probe scan, bounding the statics result
(Appendix~\ref{app:proofs}).
 \section{Conclusion}
\label{sec:discussion}

We have measured the optimization landscape of variational
peaked-circuit generation. Its statics are exact, from the
finite-depth covariance kernel to the fourth moment of the
peakedness. In the deep limit, no search over a
polynomial-parameter Lipschitz family, exhaustive search included,
beats the trap scale by more than a $\mathrm{poly}(n)$ factor on
average. At finite depth, an entanglement budget bounds the
ceiling, and the first-rung theorem gives the exact maximum at one
peaking layer. No bare fixed base fits the measured reach: the law
steepens. A better optimizer exists; it gains a few percent, and
the peakedness it reaches still falls at nearly the same rate. The
barren plateau is present, yet the second-order amplitude data are
depth-independent while the reach is not. Reachable solutions are
pairwise decorrelated yet path-connected well above the trap
scale. We synthesize the picture as shrinking reach without a
trap-scale overlap gap.

Aaronson and Zhang leave generation as an alternative between a
stalled optimizer and no efficient method \cite{aaronson2024peaked}.
Their extrapolation to experimentally relevant sizes rests on a
fixed decay base, and the measured reach does not sustain it. The
barren plateau does not explain the decay, and fragmentation of
the solution set does not either. Our hardness conjecture met its
registered falsifier when one optimizer exceeded the protocol at
the largest size. We withdraw the conjecture; the withdrawal falls
on the hardness branch, and the margin leaves the reach law
untouched. Neither branch of the alternative closes: a claim of
efficient generation must now beat a registered baseline, and a
claim of hardness must accommodate a solution set connected at the
trap scale.

The open question is the mechanism of the decay. Initialization
statistics are Haar-like, so whatever signal the optimizer
exploits is acquired during optimization, not present at the
start. The withdrawal sharpens the search: some in-class statistic
carries more signal than converged Adam extracts. On the static
side, the flat directions at the solutions are the objective's
gauge in closed form. What fixes the finite-depth ceiling is the
transverse spectrum, measured at the solutions but not modeled.
The connectivity result leaves clustering at the near-optimal
level open. Each surviving conjecture keeps its registered
falsifier. Whether the variational route to verifiable quantum
advantage stays open turns on the persistence of the reach law at
experimentally relevant sizes.
 
\section*{Data and code availability}

\begin{sloppypar}
The analysis code, its committed logs, and the figure-generating
script are archived at
\url{https://doi.org/10.5281/zenodo.22121530}; the raw restart
ensembles ($316$ array archives, $\approx 350$~MB) at
\url{https://doi.org/10.5281/zenodo.22033295}. Both records are to
be read at the version deposited with this submission, which
carries the completed $n = 16$ rows. The archived
\texttt{analysis/verify\_claims.py} recomputes $311$ of the quoted
numbers from those archives, both $n = 16$ readings, the
optimizer-class block with its architecture control, and the
converged campaign included. Its committed log reports zero
disagreements. The geometry of Sec.~\ref{sec:geometry} is covered
separately by \texttt{analysis/probe\_gauge.py}, which recomputes
every count, gap, and speed quoted there from the archived restart
ensembles. The first-rung control of Sec.~\ref{sec:atom} is
\texttt{analysis/firstrung\_inversion.py}, whose committed log
tabulates the seven ground-truth points. The digitized points of
Fig.~\ref{fig:scaling} are read
from the published raster
(\texttt{analysis/az\_fig3c\_digitize.py}), so they are rerun
against that source rather than checked by
\texttt{verify\_claims.py}. The code archive also carries the
committed-before-the-runs \texttt{REGISTRATION-CONVERGED.md}; the
restart archives of the converged grid, the
optimizer-class block, its architecture control, and the
exploratory matched-budget SGD run are in the data record. The
registered predictions rest on
the working record, as declared in Appendix~\ref{app:registration}.
\end{sloppypar}
 \section*{Acknowledgements}

This work began in Calcul Qu\'ebec's 2026 summer quantum computing
program, which proposed the peaked-circuit construction as a study
topic. The author thanks Calcul Qu\'ebec for quantum computing
training and access to computing resources during the program, and
Yuxuan Zhang for comments that sharpened the reach-law and
connectivity statements.

\section*{Author contributions}

The author designed the study, proved the theoretical results, and
executed the numerical campaigns. A large language model was used
for text drafting and editing, for analysis code, and for
verification checks.

\section*{Competing interests}

The author declares no competing interests.
 
\bibliographystyle{quantum}
\bibliography{refs}

% ------------------------------------------------------------ appendix ----
\appendix
\section{Proofs}
\label{app:proofs}

\subsection{The two-copy transfer rules}

At local dimension $d = 4$ the $S_2$ Weingarten coefficients are
$\mathrm{Wg}(e) = 1/(d^2{-}1) = 1/15$ and
$\mathrm{Wg}(s) = -1/(d(d^2{-}1)) = -1/60$
(Corollary~2.4 of Ref.~\cite{collins2006integration}, whose following
example evaluates both), and the Haar twirl of one gate reads
$\Ee[U^{\otimes 2} M U^{\dagger\otimes 2}] = \sum_{\pi, \pi'}
\mathrm{Wg}(\pi\pi'^{-1})\, \Tr[M \Pi_{\pi'}]\, \Pi_\pi$. The three
input classes of Sec.~\ref{sec:kernel} follow by evaluating
$\Tr[M\Pi_{\pi'}]$.

For the first layer, $M = (\ket{00}\bra{00})^{\otimes 2}$ is a pure
product state, so $\Tr[M\Pi_e] = \Tr[M\Pi_s] = 1$ and each output
label carries $\mathrm{Wg}(e) + \mathrm{Wg}(s) = 1/20$. For an
aligned input, $M = \Pi_e$ on the gate's two sites gives
$\Tr[M\Pi_e] = 16$, $\Tr[M\Pi_s] = 4$, hence
$16\,\mathrm{Wg}(e) + 4\,\mathrm{Wg}(s) = 1$ on $\Pi_e$ and
$16\,\mathrm{Wg}(s) + 4\,\mathrm{Wg}(e) = 0$ on $\Pi_s$: the label
passes unchanged, and the case $M = \Pi_s$ is symmetric. For a misaligned
input, $M = \Pi_e \otimes \Pi_s$ gives
$\Tr[M\Pi_e] = \Tr[M\Pi_s] = 8$, so both outputs carry
$8(\mathrm{Wg}(e) + \mathrm{Wg}(s)) = 2/5$ and the total emitted
weight is $4/5$.

The dressing $4^{\#e} 2^{\#s}$ restores stochasticity: a misaligned
input of dressed weight $4 \cdot 2 = 8$ sends $(2/5)(16/8) = 4/5$ to
$ee$ and $(2/5)(4/8) = 1/5$ to $ss$, and a first-layer gate sends
$16/20 = 4/5$ and $4/20 = 1/5$. The boundary of Eq.~\eqref{eq:kernel}
is the swap identity: with $A = \ket{w(\theta)}\bra{w(\theta)}$,
$A' = \ket{w(\theta')}\bra{w(\theta')}$ and $\Pi_c$ the copy
permutation of the final configuration, $\Tr[(A \otimes A')\,\Pi_c] =
\Tr[\rho_{S(c)}(\theta)\,\rho_{S(c)}(\theta')]$: identity sites trace
out, swapped sites pair the probes. The $k$-copy case replaces
$\{e, s\}$ by $S_k$ with Gram matrix
$G_{\pi\pi'} = 4^{\#\mathrm{cyc}(\pi\pi'^{-1})}$ and transfer
coefficients $G^{-1}t$, as stated in Sec.~\ref{sec:moments}; the
brickwall bookkeeping is that of Sec.~2, Eqs.~(16)--(18) of
Ref.~\cite{hunterjones2019unitary}.

\subsection{The truncation certificate}

The certificate is used at every quoted moment. Let
$\Ee[\delta^k] = \sum_c W(c) X_c$ be the exact configuration sum of
Sec.~\ref{sec:moments}, evaluated over a computed subset
$\mathcal{C}$.

\emph{Statement.} For every $\mathcal{C}$,
$\bigl|\Ee[\delta^k] - \sum_{c \in \mathcal{C}} W(c) X_c\bigr|
\le \sum_{c \notin \mathcal{C}} |W(c)|$.

\emph{Proof.} The omitted remainder is
$\sum_{c \notin \mathcal{C}} W(c) X_c$, and $|X_c| \le 1$, $X_c$
being an inner product between two unit vectors of
$(\mathbb{C}^D)^{\otimes k}$; the triangle inequality gives the
bound. $\square$

The right-hand side is accumulated during the sum, so the certificate
is computed rather than estimated. It is held below $1\%$ of the value
at every moment the figures plot, and the figure script prints
each exclusion; on the probe-scan points of Sec.~\ref{sec:probedep} it reaches
$2.3\%$ of $m_4$, whence $r_4$ is quoted there at $\pm 0.02$.

\subsection{Proposition~\ref{prop:ceiling}}

\emph{Setting.} $\ket{\varphi}$ is Haar-random on the unit sphere of
$\mathbb{C}^D$, $D = 2^n$. Each coordinate generates a halved Pauli
word and the derivative of the gate exponential is bounded in norm by
its generator, so $\|\partial_a w\| \le 1/2$ for every $a$ and
$\|w(\theta) - w(\theta')\| \le \|\theta - \theta'\|_1$: the probe map
is $1$-Lipschitz from $\ell_1$. Every coordinate is $2\pi$-periodic up
to a global phase that cancels in $\delta$ (Sec.~\ref{sec:gauge}), so
the maximum over $\mathbb{R}^P$ is the maximum over
$[-\pi, \pi]^P$.

\emph{Proof.} (i) \emph{Pointwise tail.} For fixed unit $w$ and Haar
$\varphi$, $\delta = |\braket{w}{\varphi}|^2$ is Beta$(1, D{-}1)$, so
$\Pr[\delta > t] = (1-t)^{D-1} \le e^{-(D-1)t}$.

(ii) \emph{Net.} A grid of spacing $2h/P$ per coordinate leaves every
$\theta$ within $\ell_1$-distance $h$ of a grid point, so by the
Lipschitz bound its image is an $h$-net of the probe family, of
cardinality $N(h) \le (\pi P/h + 1)^P$. At the radius
$h = 2^{-n/2}$ used in the body,
\begin{equation*}
\ln N \;\le\; P\big[\ln(2\pi) + \ln P + \tfrac{n}{2}\ln 2\big]
\;\le\; \tfrac{3}{2}\, P\,(n + \ln P)
\end{equation*}
for all $n, P \ge 2$.

(iii) \emph{Off-net transfer.} For any $\theta$ and its net point
$\theta_{\rm net}$, the triangle inequality and Cauchy--Schwarz with
$\|\varphi\| = 1$ give
$|\braket{w(\theta)}{\varphi}| \le
|\braket{w(\theta_{\rm net})}{\varphi}| + h$, hence
$\delta(\theta) \le 2\,\delta(\theta_{\rm net}) + 2h^2$.

(iv) \emph{Union bound.} Integrating the tail of (i) over the net,
\begin{equation*}
\Ee\big[\max_{\rm net} \delta\big] \;\le\;
\int_0^\infty \min\big(1, N e^{-(D-1)t}\big)\, dt \;=\;
\frac{\ln N + 1}{D-1}.
\end{equation*}

Combining (iii) and (iv), with $2h^2 = 2^{1-n}$ and
$1/(D-1) \le \tfrac{4}{3}\,2^{-n}$ for $n \ge 2$,
\begin{equation*}
\Ee_\varphi\big[\max_\theta \delta(\theta)\big] \;\le\;
2^{-n}\Big[4\,P\,(n + \ln P) + \tfrac{8}{3} + 2\Big] \;\le\;
7\,P\,(n + \ln P)\,2^{-n},
\end{equation*}
the last step because $P(n + \ln P) \ge 14/9$ for $n, P \ge 2$. This
is the claim, with the explicit constant $C = 7$. $\square$

\subsection{Lemma~\ref{lem:gaussian}}

\emph{Proof.} Write $v = C_1 2^{-n}$, $L = C_2 2^{-n/2}$.

(i) \emph{Phase lifting.}
$\sup_\theta |\psi(\theta)| = \sup_{(\theta,\phi) \in T}
Z_{\theta,\phi}$ with
$Z_{\theta,\phi} = \mathrm{Re}\big(e^{-i\phi}\psi(\theta)\big)$ and
$T = [-\pi, \pi]^P \times [0, 2\pi)$, a real centered Gaussian
process.

(ii) \emph{Canonical metric.} By the $L^2$ triangle inequality,
$d_Z\big((\theta,\phi),(\theta',\phi')\big) \le d(\theta,\theta') +
|\phi - \phi'|_{\circ}\sqrt{v}$, with $|\cdot|_{\circ}$ the circle
distance, at most $\pi$. Since $\Ee|\psi|^2 \le v$ everywhere,
$d(\theta,\theta') \le 2\sqrt{v}$ always, so the $d_Z$-diameter of $T$
is at most $(2+\pi)\sqrt{v}$: the integration range below is
$O(\sqrt{v})$ whatever the parameter-space diameter.

(iii) \emph{Entropy.} A grid of spacing $h/(LP)$ per coordinate in
$[-\pi, \pi]^P$ and spacing $h/\sqrt{v}$ on the phase circle leaves
every point of $T$ within $d_Z$-distance $h$ of the grid, so
\begin{equation*}
\ln N(T, d_Z, h) \;\le\;
P \ln\!\big(1 + 2\pi P L/h\big) +
\ln\!\big(1 + 2\pi\sqrt{v}/h\big).
\end{equation*}

(iv) \emph{Dudley's entropy bound} (Theorem~8.1.3 of
Ref.~\cite{vershynin2018high}):
\begin{equation*}
\Ee \sup_T Z \;\le\; C_D \int_0^{(2+\pi)\sqrt{v}}
\sqrt{\ln N(T, d_Z, h)}\; dh.
\end{equation*}
Substituting $h = \sqrt{v}\, u$ and using
$PL/\sqrt{v} = C_2 P/\sqrt{C_1}$ bounds the integral by
\begin{equation*}
\sqrt{v}\,\sqrt{P}\,\Big(\sqrt{\ln\big(2 + c\, C_2 P/\sqrt{C_1}\big)}
+ c'\Big),
\end{equation*}
with $c, c'$ universal, the $\int_0^1 \sqrt{\ln(1/u)}\, du$ tail being
a constant. For $C_1, C_2$ of order one this reads
$\Ee \sup_T Z \le C_0\sqrt{v\, P \ln(2+P)}$, constants entering
polynomially inside the logarithm.

(v) \emph{Second moment.} $\sup_T Z \ge 0$, since $T$ carries every
phase. On a separable version it is a supremum of linear functionals
of the field with coefficient norms at most $\sqrt{v}$, hence a
$\sqrt{v}$-Lipschitz function of the underlying Gaussian; Gaussian
concentration (Theorem~5.2.2 of Ref.~\cite{vershynin2018high}) gives
$\mathrm{Var}(\sup_T Z) \le C v$, so
$\Ee\big[(\sup_T Z)^2\big] \le (\Ee \sup_T Z)^2 + C v$, the constant
absorbed by $C_0$. With (iv), this gives the claim. $\square$

The true field supplies both hypotheses exactly, $C_1 = C_2 = 1$,
through Eq.~\eqref{eq:hypotheses}, so every separable Gaussian model
carrying the exact second-order amplitude data obeys
$\Ee[\sup_\theta \delta] \le C_0\, P \ln(2{+}P)\, 2^{-n}$.
The proof uses no more Gaussianity than its two tools require.
Dudley's bound asks for sub-Gaussian increments, not for a Gaussian
law. In the deep limit the true field supplies them:
$\psi(\theta) - \psi(\theta') =
\|w(\theta) - w(\theta')\|\braket{u}{\varphi}$ for a fixed unit
vector $u$, and the Beta$(1, D{-}1)$ tail established in step (i) of
the proof of Proposition~\ref{prop:ceiling} is sub-Gaussian at scale
$\|w - w'\|/\sqrt{D}$. The phase lift then carries sub-Gaussian
increments in the metric of step (ii), with an absolute constant.
L\'evy concentration of the $1$-Lipschitz map
$\varphi \mapsto \sup_\theta|\braket{w(\theta)}{\varphi}|$ on the
sphere supplies the variance step. The chaining order therefore
binds the true field in the
deep limit,
$\Ee_\varphi[\max_\theta \delta] \le C_0'\, P \ln(2{+}P)\, 2^{-n}$
with $C_0'$ not explicit; Proposition~\ref{prop:ceiling} keeps the
explicit constant.

\subsection{Theorem~\ref{thm:firstrung}}

\emph{Proof.} At $\tau_p = 1$ the peaking circuit is a single brick
layer covering, for even $\tau_r$, all $n$ qubits by disjoint pairs,
$V(\theta) = \bigotimes_{p=1}^{n/2} V_p(\theta_p)$, a tensor
product of independent gates on the disjoint brick pairs, so
$\ket{w(\theta)} = V(\theta)^\dagger \ket{0^n} =
\bigotimes_p V_p(\theta_p)^\dagger \ket{00}$ is, for every $\theta$, a
product state across the peaking layer's partition. With
$\mathcal{P}$ the set
of unit vectors of that product form,
\begin{equation*}
\delta^*(\tau_p = 1) \;=\; \sup_\theta |\braket{w(\theta)}{\varphi}|^2
\;\le\; \max_{\ket{a} \in \mathcal{P}} |\braket{a}{\varphi}|^2,
\end{equation*}
the maximum on the right attained by compactness of $\mathcal{P}$, a
product of $n/2$ unit spheres: the squared entanglement eigenvalue
$\Lambda_{\max}^2$ of $\ket{\varphi}$ for the brick product
structure, whose complement $1 - \Lambda_{\max}^2$ is the geometric
measure of entanglement of Ref.~\cite{wei2003geometric}. This
inequality is unconditional.

Under the density hypothesis each block orbit
$\{V_p(\theta_p)^\dagger\ket{00}\}$ is dense in the unit sphere of
$\mathbb{C}^4$ up to phase, so $\{\ket{w(\theta)}\}$ is dense in
$\mathcal{P}$ up to phases, which cancel in $\delta$; the supremum
of the continuous $|\braket{\,\cdot\,}{\varphi}|^2$ over a dense
subset of the compact $\mathcal{P}$ is its maximum over
$\mathcal{P}$. The inequality is then an equality. $\square$

\subsection{Proposition~\ref{prop:gauge}}

\emph{Counting.} For qubit $q$, let $\ell_q$ be the number of
layers of $V$ whose brick covers $q$; the segments on $q$ are the
$\ell_q - 1$ consecutive pairs of its gates, so
$S = \sum_q (\ell_q - 1)$. In the even-aligned open-boundary
pattern every even layer covers all $n$ qubits and every odd layer
covers $1, \ldots, n{-}2$, so each of the $n - 2$ interior qubits
carries $\tau_p - 1$ segments, and each edge qubit, covered only by
the $\lceil \tau_p/2 \rceil$ even layers, carries
$\lceil \tau_p/2 \rceil - 1 = \lfloor (\tau_p{-}1)/2 \rfloor$: the
closed form.

\emph{Invariance.} Fix a segment $(q; a, b)$ and $u \in SU(2)$
acting on $q$, embedded as $E(u)$ on the two-qubit space of either
endpoint gate. With $W$ the product of the gates between $G_a$ and
$G_b$, none touching $q$ (at the chain edge, an intervening odd
layer's gates), $E(u)$ commutes with $W$ and
\begin{equation*}
G_b\, W\, G_a
= \bigl(G_b E(u)\bigr)\, W\, \bigl(E(u)^\dagger G_a\bigr):
\end{equation*}
the dressed pair realizes the same $V$, both dressed gates having
unit determinant. Where the fifteen-rotation word is regular (rank
$15$) at each gate, the word map is a local diffeomorphism onto a
neighborhood in $SU(4)$, so for $u$ near the identity the dressing
lifts to smooth angle paths by the inverse function theorem; over
all segments this defines, on a neighborhood of the identity in
$SU(2)^S$, a smooth $\Phi$ into $(\mathbb{R}/2\pi\mathbb{Z})^P$ with
$\Phi(\mathbf{1}) = \theta$ and $V \circ \Phi$ constant.

\emph{Dimension.} A direction $(X_s)_s \in \mathfrak{su}(2)^S$
moves gate $g$ by
$\delta G_g = -\sum_{s \in \mathrm{out}(g)} E(X_s)\, G_g
+ G_g \sum_{s \in \mathrm{in}(g)} E(X_s)$, the sums over the
segments leaving and entering $g$. The word map being a local
diffeomorphism at each gate, $\mathrm{d}\Phi(X) = 0$ forces every
$\delta G_g = 0$, that is
$\sum_{\mathrm{out}(g)} E(X_s)
= G_g \bigl(\sum_{\mathrm{in}(g)} E(X_s)\bigr) G_g^\dagger$, an
equality between the local algebra
$\mathfrak{l} = \mathfrak{su}(2) \oplus \mathfrak{su}(2)$ and its
conjugate $\mathrm{Ad}(G_g)\,\mathfrak{l}$. The count
$\dim \mathfrak{l} + \dim \mathfrak{l} = 12 < 15 =
\dim \mathfrak{su}(4)$ makes a trivial intersection possible;
trivial intersection at every gate is the \emph{generic-position
condition}, and it fails on local gates. The condition is generic:
its failure is the vanishing of every $12 \times 12$ minor of the
concatenated basis matrix
$[\,\mathfrak{l} \mid \mathrm{Ad}(G_g)\,\mathfrak{l}\,]$, a
real-analytic condition on $SU(4)$. Any single transverse point
rules out identical vanishing and confines the failure to a closed
measure-zero subvariety. Under the condition both sides vanish;
the segments entering one gate sit on distinct wires, so each
$E(X_s) = 0$ separately, and every segment enters some gate. Hence
$\mathrm{d}\Phi$ is injective and
the family has dimension $3S$. The corank bound follows: the $3S$
tangent directions lie in the kernel of the projective Jacobian of
the output state, whose rank never exceeds the real dimension
$2 \cdot 2^n - 2$ of the ray space. $\square$

Regularity of the word and the generic-position condition hold off
a closed measure-zero set, and neither is assumed where it
matters: at the atom solutions and at the random points of
\texttt{analysis/gauge\_dimension.log} the corank is computed, and
equals the bound exactly.

\subsection{Proposition~\ref{prop:probegauge}}

\emph{Counting.} The closing structure sits in the last two
layers. For $\tau_p$ odd the final layer is even-aligned and
covers all $n$ qubits, so $B_2 = n/2$, $B_1 = 0$; for $\tau_p$ even it
is odd-aligned and covers $1, \ldots, n - 2$, so $B_2 = n/2 - 1$
while the two edge qubits close at their gates of the layer
before, $B_1 = 2$. No other gate carries a closed qubit.

\emph{Factorization.} A closed qubit being touched by no later
gate, $\bra{0^n} V$ factorizes at the boundary: a gate $u$ with
two closed qubits enters only through its row $\bra{00} u$, one
tensor factor of $\bra{0^n} V$, and a gate with one closed qubit
only through its $2 \times 4$ block
$(\bra{0} \otimes \mathbb{1})\, u$. A phase on any one factor is
an overall phase of $\ket{w}$, invisible to the ray.

\emph{Two local counts.} In the left-invariant frame
$\delta u = X u$, $X \in \mathfrak{su}(4)$, the row moves along
its own ray iff $\bra{00} X \propto \bra{00}$, a $u$-independent
condition since $u$ is invertible, which kills the three off-ray
entries of the first row (six real conditions) and leaves the
anti-Hermitian $3 \times 3$ block tied to the corner entry by the
trace: nine dimensions,
$9 = 15 - \dim_{\mathbb{R}} \mathbb{CP}^3$. The block moves along
its own ray iff
$(\bra{0} \otimes \mathbb{1})\, X \propto \bra{0} \otimes
\mathbb{1}$, which in $2 \times 2$ blocks forces the diagonal
upper block to $c\,\mathbb{1}$ and the off-diagonal blocks to
zero, leaving the anti-Hermitian lower block tied to $c$ by the
trace: four dimensions, $4 = 15 - 11$, eleven being the real
dimension of the Stiefel manifold $V_2(\mathbb{C}^4)$ modulo the
block phase. Both dimensions are recomputed by direct SVD in the
archived script. Word regularity lifts these local kernels to
$\theta$ directions, along with the $3S$ family of
Proposition~\ref{prop:gauge}, which fixes $V$ and a fortiori the
probe ray.

\emph{Independence.} Let a combined direction assign $X_s \in
\mathfrak{su}(2)$ to each segment and $Z_g \in \mathfrak{n}_g$ to
each closing gate, $\mathfrak{n}_g$ the local kernel above, and
suppose the total $\theta$ motion vanishes. Word regularity forces
every gate motion to vanish:
$\sum_{\mathrm{out}(g)} E(X_s) - Z_g
= \mathrm{Ad}(G_g)\bigl(\sum_{\mathrm{in}(g)} E(X_s)\bigr)$, with
$Z_g = 0$ at non-closing gates. Induct up the layers. At a gate
whose entering segments are already zero the right side vanishes;
if the gate is not closing, its outgoing insertions sit on
distinct wires, so each $X_s = 0$ separately. If it closes one
qubit, the single outgoing insertion sits on the open wire,
$E(X_s) = \mathbb{1} \otimes X_s = Z_g$, and membership of
$\mathfrak{n}_g$ forces the diagonal blocks scalar:
$X_s = c\,\mathbb{1}$ with $X_s$ traceless, so $X_s = 0$ and
$Z_g = 0$; this covers the segments whose two endpoints are both
closing gates. If it closes two qubits it has no outgoing segment
and $Z_g = 0$ outright. The base case is the first layer, which
has no entering segments. Hence the combined family is injective
at first order and the kernel has dimension at least
$3S + 9B_2 + 4B_1$. $\square$

\emph{Constant rank and exact flatness.} The construction bounds
the corank below at every regular $\theta$, so the probe-ray
Jacobian has rank at most $P - 3S - 9B_2 - 4B_1$ there. Rank is
lower semicontinuous: where the measured corank equals the bound
the rank is maximal, hence constant on a neighborhood, and the
constant-rank theorem provides adapted coordinates in which the
probe-ray map reads $(x, y) \mapsto x$. The fiber through the
point is a smooth submanifold of dimension $3S + 9B_2 + 4B_1$, and by
Eq.~\eqref{eq:field} $\delta$ is constant along it to all orders;
at a critical point the Hessian therefore carries at least that
many exact zero eigenvalues, whatever the transverse geometry
does.

Regularity of the word again holds off a closed measure-zero set;
the generic-position condition of Proposition~\ref{prop:gauge} is
not needed, the layer recursion replacing the transversality
argument. Nothing is assumed where it matters: the corank is
computed, and equals the bound, at the Haar points and the
archived best restarts of \texttt{analysis/probe\_gauge.log}.
 \section{Methods and reproducibility}
\label{app:methods}

\emph{Simulation.} Exact statevector simulation throughout
(PennyLane $0.45.1$, \texttt{default.qubit}
\cite{bergholm2018pennylane}); gates as specified in
Sec.~\ref{sec:field}; gradients by backpropagation. Instances and
restarts derive from a root seed via spawn keys
$(n, \text{instance})$ and
$(n, \text{instance}, \text{restart}, \sigma)$; each ensemble
regenerates bit-for-bit on its original architecture (optimizer
trajectories, chaotic over $400$ steps, are not bit-portable across
CPU architectures). The bottom panels of Fig.~\ref{fig:intro} are
regenerated from the archived angles and instance seed, reproducing
the archived $\delta$ values to $10^{-14}$.

\emph{Exact computations.} The two-copy kernel, the gradient
covariance, and the $S_3$/$S_4$ moment transfers are exact
summations with per-script self-tests: transfer-invariant
preservation to $10^{-12}$ (kernel) and $10^{-9}$ ($S_3$/$S_4$),
Haar limits to $10^{-9}$ (second moment) and $10^{-8}$ ($k = 4$
floor), independent Monte Carlo pipelines at small $n$, derivative
checks against fourth-order finite differences at $10^{-8}$, and
cross-evaluator agreement (dense gather versus contraction-path
evaluation, configuration by configuration). Truncated moment
evaluations report a certified error bound
$\sum |W_{\rm dropped}|$ (with $|X_c| \le 1$); its scope at every
quoted moment is stated in Appendix~\ref{app:proofs}.

\emph{Hessians and coranks.} The Hessian of $\delta$ is built from
central differences of the analytic gradient (the exact derivative
probe states), step $10^{-4}$; the asymmetry of the raw matrix,
$\sim 10^{-12}$, is its built-in error gauge. Coranks are numerical
ranks of the projected Jacobians: threshold
$10^{-8} \sigma_{\max}$ for the output state and the machine-zero
threshold $10^{-12} \sigma_{\max}$ for the probe ray, whose kernel
sits at $\sim 10^{-15}$ relative while its smallest nonzero
singular value dips to $\sim 10^{-9}$ at one shallow depth-scan
point; the bracketing singular pair is printed at every
measurement, so no threshold arbitrates a count. Solutions are the
archived best restarts by peakedness, plus one mid-band restart at
$(8, 4)$. Polishing runs L-BFGS-B on
$-\delta$ from the archived stop to $\|\nabla\delta\| < 10^{-7}$;
it is a stationarity diagnostic, outside the frozen protocol, and
no reported reach level rests on a polished value.

\emph{Campaign.} The consolidation ran on cloud spot instances as a
declarative, resumable job list with an integrity gate after each
block. A registered catch-up block later took $n = 16$ from four
instances to eighteen under the same frozen protocol, bit for bit;
per the registration, every $n = 16$ number is reported on both
sets, whichever way the change goes, and the four-instance value is
not dropped. Budget, robustness, and connectivity
protocols were fixed before the runs that tested them.
Appendix~\ref{app:registration} records their provenance tiers and
their verdicts.

\emph{The matched-budget protocol.} The unit shared by every arm of
the optimizer-class block is one value-and-gradient evaluation of
$\delta(\theta)$; for L-BFGS-B each line-search probe is charged to
the same budget, an asymmetry running against the quasi-Newton arm,
which buys fewer iterations than Adam has steps. The converged
schedule extends the frozen protocol along a doubling ladder of
caps ($400, 800, \ldots, 12{,}800$ steps) with the learning-rate
decay floored at $0.05 \times 2^{-3}$, stops when the relative gain
of a doubling falls below $0.002$, and finishes with $400$ polish
steps: a matched ceiling of $13{,}200$ evaluations, $33\times$ the
frozen budget. Appendix~\ref{app:controls} reports the control
sweep of this learning-rate floor.

\emph{The optimizer-class block.} The block closes the two-by-two
of optimizer $\in$ \{Adam, L-BFGS-B\} and initialization $\in$
\{normal, Haar\}: arms (L-BFGS-B, $\sigma$), (L-BFGS-B, Haar) and
(Adam, Haar) at $B = 16$ restarts on instances $0$--$2$ per size,
the (Adam, normal) cell being the converged grid read on the same
instances. The registered statistic is the reach ratio at matched
instances and matched budget,
$\rho_{\rm arm}(n) = R_{\rm arm}(n)/R_{\rm conv}(n)$, both sides
read at $B_0 = 16$ through the exact without-replacement estimator
$\Ee[\max \text{ of } B_0]$, the sample maximum on a $16$-restart
arm and a genuine expectation on the $32$-restart converged cell,
the same estimand at different variance. The ratio is formed per
instance and averaged, with the standard error over the three
instances; the \emph{literal} construction, mean over mean, is
computed
alongside. Appendix~\ref{app:registration} states the outcome
rules with their tie-break, and reports the architecture control.

\emph{Optimizer controls.} The control matrix of
Sec.~\ref{sec:robustness} is the frozen protocol at $n = 10$,
instance $0$, $B = 200$, in five variants: initialization scale
$0.5$ and $1.0$ (against $0.1$), learning rate $0.025$ and $0.1$
(against $0.05$), and plain SGD at the protocol's $400$ steps and
learning rate. The pacing diagnostic pairs twelve Adam and SGD
restarts from identical initializations and records the gradient
norm and parameter displacement at every step; the medians quoted
in Sec.~\ref{sec:robustness} are per-restart medians along the
trajectory, aggregated as medians over the twelve pairs. The
matched-budget SGD run quoted there is exploratory, not registered:
plain SGD under the converged schedule of this appendix
($13{,}200$-evaluation ceiling), $B = 16$, on the same instance.

\emph{The fixed-base test.} The fixed-base $p$ carried by the
introduction comes from the covariance-estimated
exact test: Hotelling's $T^2$ on the two contrasts of the
per-instance $\ln$ reach at $B_0 = 16$ on each grid, orthogonal to
an affine law in $n$, referred to
$F_{2,16}$ over the eighteen instances, $p = 0.026$ frozen and
$0.041$ converged (\texttt{analysis/\allowbreak reach\_exact\_test.log}); the
weighted $\chi^2$ of
\texttt{analysis/\allowbreak converged\_reach.log}, its
instance-mean errors treated as known, reads $0.011$ and $0.025$.
The exact-test log also decomposes the last budget doubling at
$n = 16$ per instance: nine tenths of its $0.36\%$ gain sits in a single
instance.
 \section{Convergence studies and numerical controls}
\label{app:controls}

\emph{Step-budget control.} To size what the $400$-step cap costs,
the same restart seeds were rerun at a cap of $1600$ steps, three
instances per size and sixteen restarts each ($n = 16$: one
instance). The deficit in the best reachable peakedness grows with
size,
\begin{equation*}
0.1\%,\quad 2.3\%,\quad 3.1\%,\quad 4.7\%,\quad 10.9\%
\end{equation*}
at $n = 8$--$16$ (per-size means of the per-instance best-restart
gains, which span $0.0$--$0.4\%$ at $n = 8$ and reach $6.8$,
$6.4$, $8.6$ and $10.9\%$ at $n = 10$--$16$). The median number of
steps used at the $1600$ cap is non-monotone across instances
($382$--$749$ at $n = 8$, $1226$--$1489$ at $n = 14$, $1090$ at
$n = 16$), and a residual fraction of restarts still ends at the
extended cap, $2/48$ at $n = 10$, $9/48$ at $n = 12$, $17/48$ at
$n = 14$, and $3/16$
at $n = 16$: the control is not fully converged either, so the
measured deficits are lower bounds. The deficit's magnitude matters
for comparing against Ref.~\cite{aaronson2024peaked} and for the
$n = 16$ shortfall, and Sec.~\ref{sec:steepening} quotes the raw,
the corrected, and the propagated figures. Whether its curvature
carries the steepening is settled by the registered log-linearity
test on the converged campaign's own within-trajectory deficits,
five sizes and three degrees of freedom as registered, which
rejects linearity ($\chi^2 = 16.0$, $p = 0.0011$,
\texttt{analysis/convergence\_diagnostics.log}); the
correction-based defense of the steepening is therefore withdrawn,
and Sec.~\ref{sec:steepening} rests the statement on the converged
grid directly.

\emph{Converged campaign.} The registered convergence rerun
(\texttt{REGISTRATION-CONVERGED.md}, committed before its runs)
repeats the full $18$-instance grid at $n = 8$--$16$ with
$B = 32$ restarts under the converged schedule of
Appendix~\ref{app:methods}.
Its registered acceptance criterion, a gain of at most $0.5\%$ over
the last budget doubling, is met at every size: the largest reading
is $0.107\%$ at $n \le 14$ and $0.362\%$ at $n = 16$, so the
numbers count as converged in the registered sense and no
Richardson band is published anywhere on the grid. At $n = 16$ that
last gain sits above the previous doubling's $0.288\%$, so the
registered ratio $r$ falls outside $(0, 1)$ and the band is
undefined there; per the registration the raw value and the last
gain are then reported alone. The registered statistic is the exact
expected best of $16$ restarts per instance, its error the standard
error over the eighteen instances
(\texttt{analysis/converged\_reach.log},
\texttt{analysis/convergence\_diagnostics.log}). Per the
registration, the $n = 16$ row enters the reported table and the
full-grid statistic but not the registered $n \le 14$ test.

\emph{The learning-rate floor.} The floor of the converged schedule
(Appendix~\ref{app:methods}) prevents the decay from freezing the
optimizer short of a stationary point, and its value does not
select the answer: a registered sweep over a factor $16$
($0.05 \times 2^{-2}$, $2^{-3}$ and $2^{-6}$) at $n = 12$,
instance $0$, same seeds and same rungs, reaches $0.4080940$,
$0.4080941$ and $0.4080941$, the same optimum to $10^{-7}$ at the
same median $3200$ steps
(\texttt{results/lr\_floor\_sweep/}). The optimizer-class block
therefore bears on the registered class's boundary.

\emph{The ALS restart count.} The block-ALS $\Lambda_{\max}^2$ of
Sec.~\ref{sec:atom} rests on twelve restarts per point. Raising
the count eightfold at $n = 8$, instance $0$, moves the value by
$-3.6\times10^{-16}$, leaving $0.162661$ at both counts
(\texttt{analysis/geometric\_measure\_restarts.log}). The
ground-truth side of the first-rung comparison is therefore
converged in the restart count, and the shortfalls of
Sec.~\ref{sec:atom} are the protocol's.

\emph{String-method diagnostics.} Two quantities gate the
connectivity runs. A path is \emph{valid} when the
equal-arc-length reparametrization has converged, taken as no
segment exceeding twice the mean spacing; invalid paths are
reported and excluded, and all $180$ paths behind the reported
corrugation law are valid. The \emph{adjacent-waypoint overlap},
the smallest squared overlap between consecutive waypoint states,
measures whether the polyline resolves state space: values of order
$10^{-2}$ mark a string stepping across unresolved gaps. Over those
$180$ paths it spans $0.562$--$0.935$, per-size medians $0.894$,
$0.829$, $0.773$, $0.699$ and $0.667$ at $n = 8$--$16$, so the
least-resolved path sits a factor $56$ above the unresolved regime.
Segment count is matched at $64$ across the grid, so the residual
decline of the diagnostic with $n$ is a property of the paths
rather than of the schedule. The one resolution test, $16$ to $32$
segments at $n = 8$, lowered the median $\varrho$ by $15.6\%$ while
ten of twelve paths fell and the two least-resolved rose. The top
of the grid carries the smallest diagnostic and therefore the
largest resolution uncertainty.
 \section{Registration: predictions, verdicts, and provenance}
\label{app:registration}

\subsection{Provenance tiers}

Of the campaign's protocols (Appendix~\ref{app:methods}), only
the connectivity scans' corrugation protocol, statistic, and
verdict rule sit in the archived \texttt{REGISTRATION.md}; the
moment bands below and the four predictions below
were fixed in the author's working record
before the consolidated analysis, so their antecedence rests on
that declaration, the public archive postdating the runs. The
optimizer-class block of Sec.~\ref{sec:robustness} is registered
differently: its arms, statistic, and the cost of each outcome are
fixed in the archived \texttt{REGISTRATION-CONVERGED.md},
committed before any of its runs existed, so its antecedence rests
on the repository history rather than a declaration. One archived
output, \texttt{analysis/campaign\_verdicts.log}, predates the
readings of Sec.~\ref{sec:shelf} and the verdicts below and
diverges from them twice. It grades prediction 1 confirmed on the
separate-grid reference whose ratios Sec.~\ref{sec:budget} reports
in parentheses, $0.69$ and $0.75$, where the same-ensemble
reference gives $0.63$ and $0.61$ and
the verdict below counts that fall as failing the prediction
as stated; its printed label for that reference reads $[6, 200]$,
while the slopes it cites are fitted on $[25, 200]$
(\texttt{analysis/budget\_scan.py}). And it carries the corrugation
rate clause on the asymptotic $p$ that Sec.~\ref{sec:shelf}
calibrates. The manuscript supersedes it on both.

\subsection{The moment bands}

The record carries the fourth-moment exponent band of
Sec.~\ref{sec:moments}, $\ln m_4/\ln m_2 \in [4.5, 6)$ with the
upper half favored; both measured values land inside it, the
$n = 8$ one in the favored half. A companion registration on the
third moment at
$(n, \tau_r) = (8, 8)$ carried two bands, $m_3 \in [2.0, 2.35]$ and $\ln m_3/\ln m_2 \in [2.5, 3]$. The
numeric band was missed, the measured $m_3 = 1.549$ falling below
it. The miss traces to a units error. Both bands are
Eq.~\eqref{eq:envelope} at $k = 3$, $m_3 = m_2^3$. The numeric one
evaluated that cube at the variance ratio $1.33$ instead of at
$m_2$ itself, the flat mean tying the two exactly
($\mathrm{Var}(\delta)/D^{-2} = 2m_2 - 1$). Converted, the same
registered number gives $m_2 = 1.165$ and predicts $1.58$, within
$2.1\%$ of the measurement. The band counts as failed nonetheless,
the correction postdating the measurement. Meeting the exponent
band at $2.88$ offsets nothing, the two bands not being
independent.

\subsection{The four predictions and their verdicts}

Before the consolidated campaign ran, four predictions with
quantitative falsifiers were registered against it.

\begin{enumerate}
\item \emph{Budget-800.} $R(n, B)$ keeps its measured log-budget
slope to $B = 800$; a re-emerging saturating atom at $n = 12$
falsifies C-reach. \textbf{Falsified as stated; falsifier not
triggered, so C-reach stands}: no degenerate maximum re-emerges at
$B = 800$, and the slope falls to $0.61$ to $0.63$ of its own
$[25, 200]$ reference over the same ensembles, still
$\approx 9\times$ the i.i.d.-exponential null at $n = 12$
(Sec.~\ref{sec:budget}).
\item \emph{The $n = 16$ point.} The mean-of-best continues the fixed
$1.195^{-n}$ law within consolidated errors. \textbf{Falsified as
stated}, superseded by a steepening law (Sec.~\ref{sec:steepening}).
At fixed budget this is what $\ln B = \Omega(n)$ would produce: a
stronger form of hardness than the one registered.
\item \emph{Robustness.} The profile is optimizer-independent within
the stable-local class. \textbf{Falsified as registered}: the class
as registered carries no pacing criterion, plain SGD is inside it,
and at the protocol's step budget its profile is not the protocol's
(Sec.~\ref{sec:robustness}). The grade
follows the standard of the moment bands above:
the pacing criterion that would excuse SGD postdates the SGD run and
buys no grade here. Confirmed for hyperparameter variation of the
frozen protocol (initialization scale $\times 10$, learning rate
$\times 4$, one instance at $n = 10$). Two measurements bound what
survives: paired trajectories locate SGD's failure in its pacing, not
in the landscape, and the preregistered optimizer-class block of
Sec.~\ref{sec:robustness} tests the boundary upward, across distinct
optimizers the registered class contains. That block returns
$\rho_{\rm arm} \le 1$ within errors through $n = 14$ and an
exceedance at $n = 16$ for (L-BFGS-B, $\sigma$)
(Table~\ref{tab:optclass}): the prediction fails in both
directions, downward by SGD's pacing at the frozen budget, upward by
a quasi-Newton member exceeding converged Adam at the largest size.
The upward exceedance triggers the third falsifier of C-hardness and
withdraws it (Sec.~\ref{sec:hardness}).
\item \emph{Corrugation scaling.} $\varrho(n)$ extrapolates the
$0.68 \to 0.5$ trend of the first-pass strings polynomially
(instance $0$: the minimum best-path ratio at $n = 8$ and the
resolved-band midpoint at $n = 12$, two statistics the consolidated
design no longer measures). \textbf{Falsified as stated}: on the
registered statistic of the consolidated sweep the floor falls to
$0.226$ at $n = 16$, far below that trend, and connectivity itself
persists (Sec.~\ref{sec:shelf}).
\end{enumerate}

\subsection{The optimizer-class outcome rules}

The rule that the costlier of the two
readings is reported was fixed in the analysis code
(\texttt{analysis/\allowbreak optclass\_reach.py}, commit
\texttt{c883ea6},
2026-08-09) before the $n = 16$ archives invoking it existed
(2026-08-10/11): the git history is the timestamp. The
optimizer-class block (Appendix~\ref{app:methods}) is the first
place the two
constructions part company, the paired statistic triggering the
growing-excess branch and the literal one not, and the costlier
reading is the one reported
(Table~\ref{tab:optclass}). The three outcomes and their costs were
fixed in \texttt{REGISTRATION-CONVERGED.md} before the data
existed: $\rho \le 1$ within errors at every size (the ceiling is
the class's); $\rho > 1$ at two standard errors at some size but
flat in $n$ (every absolute reach value is relabeled as Adam's and
the quantitative comparison with the fitted base of
Ref.~\cite{aaronson2024peaked} is withdrawn); or $\rho > 1$ and
growing with $n$ (the third registered falsifier of C-hardness is
triggered by our own measurement and C-hardness is withdrawn, not
weakened). The measured outcome is the third
(Secs.~\ref{sec:dynamics} and~\ref{sec:hardness}). The registered
threshold is two standard errors, applied as registered; what the
growth criterion is worth, and the statistical weight of the
deciding cell, are stated with the measurement in
Sec.~\ref{sec:robustness}, whose per-instance ratios are in
\texttt{analysis/optclass\_reach.log}.

\subsection{The architecture control}

The \texttt{ENV} files record a machine-type substitution within
the registered architecture, \texttt{n2-standard-96} for
\texttt{n2-highcpu-96}. The provenance deviation that matters is
disclosed and controlled rather than assumed harmless. The $n = 8$--$14$ arms ran on Apple silicon and
the $n = 16$ cells on the campaign's x86 machines, the converged
denominator being x86 throughout (\texttt{ENV} files in
\texttt{results/optclass/}): the one size whose ratio is formed
within a single architecture is the size that decides the verdict,
and the switch falls between $n = 14$ and $n = 16$, where $\rho$
moves most. A dedicated control measures the switch instead of
assuming it away: the deciding arm recomputed on x86 at $n = 12$
under identical settings, its reading rule committed before the
control ran (\texttt{cloud/runner.py}, commit \texttt{055cb31}).
The two runs produce distinct trajectories, per-restart weights
differing by up to $1.2\times 10^{-2}$ with different step counts,
yet the per-instance bests coincide to six decimals:
$\rho_{\rm x86}(12) = 0.9988 \pm 0.0058$ against
$\rho_{\rm ARM}(12) = 0.9988 \pm 0.0058$. The two figures are equal
because the optima are, and $z = 0.00$ records that coincidence
rather than a powered test; the committed rule ($z < 2$) closes the
confound. The architecture
moves the paths, not the optima, the endpoint degeneracy of
Sec.~\ref{sec:shelf} doing the controlling, so the confound between
architecture and size does not carry the optimizer-class verdict.
The registration file is not amended.
 \section{The cube-root conjecture at fixed depth}
\label{app:cuberoot}

We test the instance-exact form of the cube-root conjecture of
Ref.~\cite{aaronson2024peaked} (their Conj.~3.2):
$\delta^* \simeq b_{\rm triv}^{\gamma}$ with $\gamma$
independent of the peaking depth, where $b_{\rm triv}$ is the
peakedness of the same instance truncated to its first
$\tau_r - \tau_p$ random layers. Their conjecture leaves the
exponent a free constant. Their text reports it empirically as almost
a cube root, and $\gamma = 1/3$ is the value we test.
Ref.~\cite{aaronson2024peaked}
states the conjecture asymptotically, in the mean, and at fixed ratio
$\tau_p/\tau_r$, a ratio the depth scan varies, so what follows
constrains this per-instance, fixed-depth reading rather than the
original claim. On our grid the fitted exponent
$\gamma = \ln \delta^* / \ln b_{\rm triv}$ is not depth-independent:
as a mean over three instances it falls
$0.28 \to 0.17 \to 0.08 \to 0.055$ over $\tau_p = 1$--$4$ at $n = 8$,
and $0.19 \to 0.11 \to 0.035$ over $\tau_p = 1$--$3$ at $n = 6$. The
per-instance values at $(n, \tau_p) = (8, 1)$ are $0.215$, $0.327$ and
$0.292$: one instance sits at $\approx 1/3$ before the fall, so these
data exclude depth-independence of $\gamma$ at these sizes, not the
value $1/3$ at $\tau_p = 1$. Two sizes and four depths cannot speak to
an asymptotic statement.
 
\end{document}